\documentclass[11pt]{article}
\usepackage{lmodern}
\usepackage[utf8]{inputenc}
\usepackage{color}
\usepackage{graphicx}
\usepackage{xspace}
\usepackage{multirow}
\usepackage{url}
\usepackage{amsmath,amssymb,amsfonts,amsthm}
\usepackage{microtype,mathtools,mathrsfs}
\usepackage{times}
\usepackage{hyperref}
\hypersetup{pdfstartview=XYZ}

\renewcommand{\baselinestretch}{1.2}
\usepackage{amsmath,amssymb,amsfonts}
\usepackage[lined]{algorithm2e}

\usepackage{multirow}
\usepackage{microtype,mathtools,mathrsfs}

\newtheorem{theorem}{Theorem}[section]

\newtheorem{claim}[theorem]{Claim}
\newtheorem{assert}{Assertion}
\newtheorem{lemma}[theorem]{Lemma}

\newtheorem{corollary}[theorem]{Corollary}

\newtheorem{remark}[theorem]{Remark}

\newtheorem{observation}[theorem]{Observation}

\newcommand\EE{\mathbb{E}}

\newcommand{\bigO}{\mathcal O}

\newcommand{\ADV}{\mbox{\tt Adv}}

\newcommand{\uniform}{\mbox{\rm uniform}}
\newcommand{\kk}{\mbox{\rm \#k}}

\newcommand{\Natural}[0]{\mathbf{N}}

\newcommand\PP{\mathbb{P}}

\def\cO{{\cal O}}
\def\cT{{\cal T}}
\def\cA{{\cal A}}

\def\cP{{\cal P}}

\newenvironment{smallitemize} {
  \begin{list}{$\bullet$} {\setlength{\parsep}{0pt}
\setlength{\itemsep}{0pt}} } { \end{list} }

\usepackage{times}

\begin{document}
\begin{titlepage}

\title{\bf The ANTS Problem\footnote{Preliminary results of this paper appeared in the proceedings of the 31st Annual ACM SIGACT-SIGOPS Symposium on
Principles of Distributed Computing (PODC), 2012, and the 26th international symposium on Distributed Computing (DISC)
2012, as part of \cite{FK12,FKLS}. 
 O.F., incumbent of the Shlomo and Michla Tomarin Career Development Chair,  was partially supported by the Clore Foundation and the Israel Science Foundation (grant 1694/10). A.K. was partially supported by the ANR projects DISPLEXITY and PROSE, and by the INRIA project GANG.
This project has received funding from the European Research Council (ERC) under the European Union's Horizon 2020 research and innovation programme  (grant agreement No 648032).}}
\author{Ofer Feinerman\\Weizmann Institute of Science\\ {\small\tt feinermanofer@gmail.com} \and
Amos Korman\\ CNRS and Univ. Paris Diderot\\ {\small\tt amos.korman@liafa.univ-paris-diderot.fr}}

\date{}

\maketitle

\begin{abstract}
We introduce the \emph{Ants Nearby Treasure Search (ANTS)} problem, which models natural cooperative foraging behavior such as that performed by ants around their nest. In this problem,  $k$~probabilistic agents, initially placed at a central location, collectively search for a treasure on the two-dimensional grid. The treasure is placed at a target location by an adversary and the agents' goal is to find it as fast as possible as a function of both $k$ and $D$, where $D$ is the (unknown) distance between the central location and the target. We concentrate on the case in which agents cannot communicate while searching. It is straightforward to see that the time until at least one agent finds the target is at least $\Omega(D+D^2/k)$, even for very sophisticated agents, with unrestricted memory. Our algorithmic analysis aims at establishing connections between the time complexity and the initial knowledge held by agents  ({\em e.g.}, regarding their total number~$k$), as they commence the search. We provide a range of both upper and lower bounds for the initial knowledge required for obtaining fast running time. 
For example, we prove that $\log \log k+\Theta(1)$ bits of initial information are both necessary and sufficient to obtain 
asymptotically optimal running time, {\em i.e.,} $\bigO(D+D^2/k)$. We also we prove that  for every $0<\epsilon< 1$, running in time $\bigO(\log^{1-\epsilon} k\cdot (D+D^2/k))$ requires that agents have the capacity for storing $\Omega(\log^\epsilon k)$ different states as they leave the nest to start the search.  To the best of our knowledge, the lower bounds presented in this paper provide the first non-trivial lower bounds on the memory complexity of probabilistic agents in the context of search problems.

We view this paper as a ``proof of concept'' for a new type of interdisciplinary methodology. To fully demonstrate this methodology, the theoretical tradeoff presented here (or a similar one) should be combined with measurements of the time performance of searching ants.

\bigskip

\end{abstract}

\date{}
\end{titlepage}

\section{Introduction}\label{sec:introduction}

The universality of search behavior is reflected in multitudes of studies in different fields including computer science, robotics,  and biology. We use tools from distributed computing to study a biologically inspired scenario in which multiple agents, initially located at one central location, cooperatively search for a treasure in the plane. The goal of the search is to locate nearby treasures as fast as possible and at a rate that scales well with the number of participating agents.

Many animal species search for food around a central location that serves as the search's initial point, final destination or both~\cite{OrPe79}. This central location could be a food storage area, a nest where offspring are reared or simply a sheltered or familiar environment. Central place foraging holds a strong preference to locating nearby food sources before those that are further away. Possible reasons for that are, for example: (1) decreasing predation risk~\cite{Kre80}, (2) increasing the rate of food collection once a large quantity of food is found~\cite{OrPe79}, (3) holding a territory without the need to reclaim it~\cite{GKDL94,Kre80}, and (4) facilitating the navigating back after collecting the food using familiar landmarks~\cite{CDGW92}.

Searching in groups can increase foraging efficiency~\cite[p.~732]{ant-book}. In some
extreme cases, food is so scarce that group searching is believed to be
required for survival~\cite{CCG99,JBS98}. Proximity of the food source to the
central location is again important in this case. For example, in the case of
recruitment, a nearby food source would be beneficial not only to the
individual that located the source but also increase the subsequent
retrieval rate for many other collaborators~\cite{Tra77}. Foraging in groups
can also facilitate the defense of larger territories~\cite{Sch71}. 

Eusocial insects ({\em e.g.,} bees and ants) often live in a
single nest or hive, which naturally makes their foraging patterns central.
They further engage in highly cooperative foraging, which can
be expected as these insects reduce competition between individuals to a
minimum and share any food that is found. In many cases, this collaborative effort is done with hardly any communication between searchers. For example,  mid-search communication of  desert ants {\em Cataglyphys} and honeybees {\em Apis mellifera}  is highly limited due to their dispersion and lack of chemical trail markings \cite{Razin}. Conversely, it is important to note that these insects do have an opportunity to interact amongst themselves before they leave their nest as well as some capacity to assess their own number \cite{Pratt05}. 

Both desert ants and bees possess many individual navigation skills. These include the capacity to maintain a compass-directed
vector flight~\cite{CSO+00,HM85}, measure distance using an internal
odometer~\cite{SW04,SZAT00}, travel to distances taken from a random power law
distribution~\cite{RSR}, and engage in spiral or quasi-spiral movement
patterns~\cite{R08,RSMG07,WS81}.  Finally, the search trajectories of desert ants have been shown to  include two distinguishable sections:  a long straight path in a given direction emanating from the nest and a second more tortuous path within a small confined area~\cite{HM85,WMZ04}.

\paragraph{The ANTS Problem.}
In this paper, we theoretically address general questions of collective central place searching.
More precisely, we introduce the \emph{Ants Nearby Treasure Search (ANTS)} problem, a generalization of the  \emph{cow-path} problem \cite{DFG06,KRT96,KSW86}, in which $k$ identical (probabilistic) agents, initially placed at some central location,
collectively search for a treasure in the two-dimensional grid (modeling the plane). The treasure is placed by an adversary at some target location at distance $D$ from the central location, where $D$ is unknown to the agents. Traversing an edge of the grid requires one unit of time. The goal of the agents is to find the treasure as fast as possible, where the time complexity is evaluated as a function of both $k$ and $D$.

In this distributed setting, we are concerned with the
\emph{speed-up} measure (see also,~\cite{AB96, AB97,FIPS11,FGKP06}), which aims to capture  the impact of using $k$ searchers in comparison to using a single one. Note that the objectives of quickly finding nearby treasures and having significant speed-up may be at conflict. That is, in order to ensure that nearby treasures are quickly found, a large enough fraction of the search force must be deployed near the central location. In turn, this crowding can potentially lead to overlapping searches that decrease individual efficiency.

 It is a rather straightforward observation that the  minimal time required for finding the treasure is $\Omega(D+ D^2/k)$.
Indeed, one must spend $D$ time merely to walk up to the treasure, and, intuitively, even if the agents divide the domain perfectly among themselves, it will still require them $\Omega(D^2/k)$ steps to explore all points at distance~$D$. Our focus is on the question of how  agents can approach this bound if their communication  is limited or even completely absent.  Specifically, we assume that no communication can be employed between agents during the execution of the algorithm, but, on the other hand, communication is almost unlimited when performed in a preliminary stage, before the agents start their search. The former assumption is inspired by  empirical evidence that desert ants (such as {\em Cataglyphys Niger} or {\em  Cataglyphis savignyi}) rarely communicate outside their nest \cite{HM85}, and the latter assumption is very liberal, aiming to cover many possible scenarios.

 To simulate the initial step of information sharing  within the nest  we use the abstract framework of {\em advice}.
 That is, we model the preliminary process for gaining knowledge  by means of an  {\em oracle} that assigns advice  to agents. To measure the amount of information accessible to agents, we analyze the {\em advice size}, namely, the maximum number of bits used in an advice\footnote{Note that a lower bound of $x\geq 1$ bits on the advice size implies a lower bound of $2^{x-1}$ on the number of states used by the agents already when commencing the search.}. Since we are mainly interested in lower bounds on the advice size required to achieve a given time performance, we apply a liberal
  approach and assume a highly powerful oracle. More specifically, even though it is supposed to model a distributed probabilistic process, we assume that the oracle is a centralized probabilistic algorithm (almost unlimited in its computational power) that can assign different agents with different advices.  Note that, in particular, by considering identifiers as part of the advice, our model allows to relax the assumption that all agents are identical and to allow  agents to be of  several types. Indeed, in the context of ants, it has been established that on their first foraging bouts  ants execute different  protocols than those executed by more experienced ants~\cite{WMZ04}.

\subsection{Our Results}

We use tools from distributed computing to study a biologically inspired scenario which models  cooperative foraging behavior such as that performed by ants around their nest. More specifically, we introduce the ANTS problem, a natural problem inspired by   desert ants foraging around their nest, and study it in a simple, yet relatively realistic,  model.

\paragraph{Technical Contribution.} We establish connections between the advice size and the running time\footnote{To measure the time complexity, we adopt the terminology of competitive analysis, and say that an algorithm is $c$-competitive, if its time complexity is at most $c$ times the straightforward lower bound, that is, at most $c(D+ D^2/k)$. In particular, an $\bigO(1)$-competitive algorithm is an algorithm that runs in $\bigO(D+ D^2/k)$ time.}. In particular, to the best of our knowledge, our lower bounds on the advice size consist of the first non-trivial lower bounds on the memory of probabilistic searchers.

On the one extreme, we have the case of  \emph{uniform} search algorithms, in which
the agents are not assumed to have any a-priori information, that is, the advice size is zero.
We completely characterize the speed-up penalty that must be paid when using uniform algorithms.  Specifically,
we show that for a function $\Phi$ such that $\sum_{j=1}^{\infty} 1/\Phi(2^j)$ converges, there
exists a uniform search algorithm that is $\Phi(k)$-competitive.
On the other hand,  we show that if $\sum_{j=1}^{\infty} 1/\Phi(2^j)$ diverges, then
 there is no uniform search algorithm that is $\Phi(k)$-competitive.
In particular, we obtain the following theorem, implying that the penalty for using uniform algorithms is  slightly more than logarithmic in the number of agents.

\bigskip

\begin{theorem}\label{thm:upper-uniform}
For every constant $\varepsilon> 0$,  there exists a uniform search algorithm that is $\bigO(\log^{1+\varepsilon} k)$-competitive, but there is no uniform search algorithm that is $\bigO(\log k)$-competitive.
\end{theorem}
\noindent In fact, $\bigO(\log k)$-competitive algorithms not only require some advice, but actually require $\log\log\log k+\Theta(1)$ bits of advice. That is expressed in the following theorem.\bigskip
\begin{theorem}\label{thm:upper-log}
There is an $\bigO(\log k)$-competitive search algorithm using advice size of $\log\log\log k+\bigO(1)$ bits. On the other hand, any  $\bigO(\log k)$-competitive  algorithm has advice size of $\log\log\log k+\Omega(1)$ bits.
\end{theorem}

\noindent On the other extreme, we show that advice of size $\log\log k+\Theta(1)$ bits is both necessary and sufficient to obtain an asymptotically optimal running time. \bigskip

\begin{theorem}\label{thm:tight}
There exists an $\bigO(1)$-competitive search algorithm using advice of size $\log\log k+\bigO(1)$. This bound is tight as any $\bigO(1)$-competitive search algorithm must have advice size of $\log\log k-\Omega(1)$ bits.
\end{theorem}

\noindent We further exhibit  lower and upper bounds for the  advice size for a range of  intermediate competitivenesses. Specifically, the following lower bound on the advice size implies that advice of size $\Omega(\log\log k)$ is necessary
even for achieving much larger time-competitiveness than constant.
\bigskip

\begin{theorem}
Fix~$0<\epsilon\leq 1$. Any search algorithm that is $\bigO(\log^{1-\epsilon} k)$-competitive  must have advice size of $\epsilon\log\log k-\Omega(1)$ bits.
\end{theorem}
\bigskip
\begin{theorem}
Fix~$0<\epsilon\leq 1$. Any $\bigO\left(\frac{\log k}{2^{\log^\epsilon\log k}}\right)$-competitive search algorithm requires  advice size of $\log^{\epsilon}\log k-\bigO(1)$. 
\end{theorem}
\medskip
\noindent Complementing this result, we show that the lower bound on the advice size is asymptotically optimal.\\ 
\begin{theorem}
Fix~$0<\epsilon\leq 1$. There exists an $\bigO\left(\frac{\log k}{2^{\log^\epsilon\log k}}\right)$-competitive search algorithm using  advice size of $\bigO(\log^{\epsilon}\log k)$. 
\end{theorem}

\medskip
\noindent  The aforementioned lower and upper bounds are summarized in Table 1.

{\renewcommand{\arraystretch}{0.8}
\renewcommand{\tabcolsep}{0.2cm}
\begin{table*}[h]\label{tab:results}
\begin{center}
{
\bgroup\large
\begin{tabular}{|l|l|l|}
\hline
& Competitiveness
& Advice size    \\
\hline
Upper bound  & $\bigO(\log^{1+\epsilon} k)$ & zero \\
\hline
Lower bound  & $\omega(\log k)$ & zero \\
\hline
{Tight bound }
 & $\bigO(1)$    & $\log\log k+\Theta(1)$\\

\hline
{Lower bound }
 & $\bigO(\log^{1-\epsilon}  k)~~~~~~~0<\epsilon < 1$    & $\epsilon\log\log k-\Omega(1)$\\
\hline
Tight bound  & $\bigO(\log k/ 2^{\log^\epsilon\log k})~~~0<\epsilon < 1$  & $\Theta(\log^{\epsilon}\log k)$ \\

\hline
Tight bound  & $\bigO(\log k)$ & $\log \log\log  k+ \Theta(1)$\\

\hline
\end{tabular}
\egroup
}
\end{center}
\caption{Bounds on the advice size for given competitiveness}
\end{table*}

Apart from their purely theoretical appeal, the algorithms we suggest in the aforementioned upper bound results are used mainly to indicate the tightness of our lower bounds. Indeed, these particular algorithms use somewhat non-trivial components, such as being executed in repeated iterations, which seem to correspond less to realistic scenarios. As our final technical result, we  propose a uniform search algorithm that is concurrently efficient
and extremely simple and, as such, it may imply some relevance for actual biological scenarios. 

\paragraph{Conceptual Contribution.}
This paper may serve as a ``proof of concept'' for a new bridge  connecting computer science and biology. Specifically, we illustrate that  theoretical distributed computing can potentially provide a novel and efficient methodology for the study of highly complex, cooperative biological ensembles.
To this end, we choose a setting that is, on the one hand, sufficiently simple to be analyzed theoretically, and, on the other hand,  sufficiently realistic so that this setting, or a variant of it, can be realizable in an experiment. Ultimately, if an experiment that
complies with our setting reveals that the ants' search is time efficient, in the sense detailed above, then our theoretical lower bound results would imply lower bounds on the number of states ants have already when they exit their nest to commence the search. Note that lower bounds consisting of even a few states would already demonstrate the power of this tradeoff-based methodology, since they would be very difficult to obtain using more traditional methodologies.

\subsection{Outline}
In Section \ref{sec:preliminaries}  we formally describe the ANTS problem and provide basic definitions.
To make the reader more familiar with the model, we start the technical part of the paper
in Section \ref{sec:upper-opt} by providing a relatively simple $\bigO(1)$-competitive search algorithm. This algorithm uses a-priori knowledge regarding the total number of agents. Specifically, the algorithm assumes that each agent knows a $2$-approximation of~$k$,  encoded as advice of size $\log\log k$. In Section \ref{sec:lower} we turn to prove lower bounds on the advice size. We start this section by proving that the competitiveness of any uniform search algorithm is more than logarithmic. 
This result is a special case of the main lower bound theorem, namely, Theorem \ref{main-lower}, that we prove just next. The reason we prove this special case separately is because the proof is simpler, and yet, provides some of the main ingredients of the more general lower bound proof.  We then describe specific lower bounds  that directly follow from Theorem \ref{main-lower} as corollaries.
We conclude the technical part of the paper in Section \ref{sec:upper-bound} where we provide several constructions of search algorithms. We first present several constructions whose performances provide  tight upper bounds  for some of our lower bounds results.
We then turn our attention to present a relatively efficient uniform algorithm, called, the {\em harmonic} algorithm, whose simplicity   may suggest its relevance to actual search by ants. In Section \ref{sec:conclusion} we discuss the conceptual
contribution of the paper as a proof of concept for a new interdisciplinary methodology that combines algorithmic studies with biology experiments. 


\subsection{Related Work}
Our work falls within the framework of natural algorithms, a recent attempt to
study biological phenomena from an algorithmic
perspective~\cite{Afek-DISC,AAB+11, BMV12, Cha09,emek,FKLS}.

Since our preliminary conference publications in  \cite{FK12,FKLS}, the ANTS problem had quickly attracted attention from the community of distributed computing, with  continuation works by other researchers, studying various theoretical aspects of our model \cite{Uitto,Emek-ants,STOC16,OPODIS16,Tobias,Nancy-PODC}. Specifically, \cite{Uitto,Emek-ants,Tobias} studied a variant of our model where
agents are modeled as finite state machines and can communicate outside the nest throughout the execution. In particular, under these assumptions, Emek et al., showed that agents can solve the ANTS problem in $\bigO(D+D^2/k)$ time  \cite{Uitto}. This implies that in this case, no a-priori knowledge about $k$ is required by the agents. Note, however, that this assumption regarding communication is probably not applicable for the case of desert ants, where no communication by such searching ants has been recorded.
The ANTS problem was further studied  by Lenzen et al.,  who investigated the effects
of bounding the memory as well as the range of available probabilities of the
agents \cite{Nancy-PODC}. Fraigniaud et al. considered the case where $k$ non-communicating agents search on line for an adversarially placed treasure \cite{STOC16}. They also argued for the usefulness of parallel search performed by non-communicating agents due to the inherent robustness of such computation. A continuation work considering the case where the treasure is placed uniformly at random in a large domain was studied in \cite{OPODIS16}.  \\

The ANTS problem is a special case of collaborative  search -- a classical family of problems that has been extensively studied in different fields of science. We next review previous works on search problems, organizing them according to the research disciplines in which they were investigated.

\paragraph{Search Problems in Computer Science.}
In the theory of computer science,
the exploration of graphs using mobile agents (or robots) is a central question. (For a more detailed survey refer to {\em e.g.,}~\cite{FKLS,FDPS10}.) When it comes to probabilistic searching, the random
walk is a natural candidate, as  it is extremely simple, uses no memory, and
trivially self-stabilizes. Unfortunately, however, the random walk turns out to be inefficient
in a two-dimensional infinite grid. Specifically, in this case, the expected hitting time is infinite,
even if the treasure is nearby.

Most graph exploration research is concerned  with the case
of a single deterministic  agent exploring a finite graph, see for example \cite{AH00,DP02, GPRZ07, PP99, Rein08}.
In general, one of the main challenges in search problems is the establishment of memory  bounds. For example, the question of whether a single agent can explore all finite undirected graphs using logarithmic memory was open for a long time; answering it to the affirmative~\cite{Rein08} established an equality between the classes of languages SL and L. As another example, it was proved in \cite{Ro08} that no finite set of constant memory agents  can explore
all graphs. To the best of our knowledge, the current paper is the first paper establishing non-trivial memory lower bounds in the context of randomized search problems.

In general, the more complex setting of using multiple identical agents has received much less attention.
Exploration by deterministic multiple agents was studied in, {\em e.g.,}~\cite{FIPS11,FGKP06}.
To obtain better results when using several identical deterministic agents, one must assume that the agents are either centrally coordinated or
that they have some means of  communication (either explicitly, or implicitly, by being able to detect the presence of nearby agents).
When it comes to probabilistic agents,
analyzing the speed-up measure for  $k$ random walkers
has recently gained attention.
In a series of papers,  a speed-up of $\Omega(k)$ is established for various finite graph families,
including, {\em e.g.,} expanders and random graphs~\cite{AAKKLT, ES11, CFR09}.
In contrast, for the two-dimensional $n$-node grid, as long as $k$ is polynomial in $n$, the speed up is only logarithmic in~$k$.
The situation with infinite grids is even  worse.
Specifically, though the $k$ random walkers find the treasure with probability one, the expected (hitting) time  is infinite.

 Evaluating the running time as a function of $D$, the distance to the treasure, was studied in the context of the {\em cow-path} problem \cite{BCR91}. This problem considers a single mobile agent placed on a graph, aiming to find a treasure placed by an adversary. It was established in \cite{BCR91} that the competitive ratio for deterministically finding~a~point on the real line is nine, and that in the two-dimensional grid,  the spiral search algorithm is optimal up to lower order terms. Several other variants of this problem were studied in  \cite{DFG06,KRT96, KSW86, LS01}.
 In particular, in~\cite{LS01},  the cow-path problem was extended by considering $k$ agents. However, in contrast to our setting, the agents they consider have unique identities, and the goal is achieved by (centrally) specifying a different path for each of the $k$ agents.

The notion of advice is central in computer science. In particular,
the concept of advice  and its impact on various computations
has recently found various applications in distributed computing. In this context, the main measure used is the advice size.
 It is for instance analyzed in frameworks such as non-deterministic decision \cite{FKP11,KK07,KKP10}, broadcast \cite{FIP06}, local computation of MST \cite{FKL10}, graph coloring  \cite{FGIP07}
and  graph searching by a single robot \cite{CFIKP08}. Very recently, it has also been investigated in the context of online algorithms \cite{BKKK11, EFKR11}.

The question of how important it is for individual processors to know their total number has recently been addressed in the context of locality.  Generally speaking, it has been observed that for several
classical local computation tasks, knowing the number of processors is not essential~\cite{KSV11}. On the other hand, in the context of  local distributed decision, some evidence exist that such knowledge is crucial for non-deterministic verification~\cite{FKP11}.

\paragraph{Search Problems in Biology and Physics.}
Group living and food sources that have to be
actively sought after make collective foraging a widespread biological
phenomenon.  Social foraging theory~\cite{GC00} makes use of economic and
game theory to optimize food exploitation as a function of the group size and
the degree of cooperation between agents  in different environmental
settings. This  theory has been extensively compared to experimental
data (see, {\em e.g.,}~\cite{AMFL10,GiGi88}) but does not typically account for the
spatial characteristics of resource abundance. Central place foraging
theory~\cite{OrPe79}  assumes a situation in which food is collected from a
patchy resource and is returned to a particular location, such as a nest. This theory is used to calculate optimal durations for exploiting food patches at different distances from  the central location and has also been tested against experimental observations~\cite{GKDL94,HoPo87}. Collective foraging around a central location is particularly interesting in the case of social insects where large groups forage cooperatively with, practically, no competition between individuals.

 Actual collective search trajectories of non-communicating agents have been studied in the physics literature ({\em e.g.,}  \cite{R06,BH97}). Reynolds \cite{R06} achieves optimal speed up through overlap reduction which is obtained by sending searchers on near-straight disjoint lines to infinity. This must come at the expense of finding proximal treasures. Harkness and Maroudas
 \cite{HM85} combined field experiments with computer simulations of a semi-random collective search and suggest substantial speed ups as  group size increases.

\paragraph{Search Problems  in Robotics.}

From an engineering perspective,
the distributed cooperation of a team of autonomous agents (often referred to
as robots or UAVs: Unmanned Aerial Vehicles) is a problem that has been extensively studied. These models extend single agent searches in which an agent  with limited sensing abilities attempts to locate one or several mobile or immobile targets~\cite{PYP01}. The memory and computational capacities of the agent are typically large and many algorithms rely on the construction of cognitive maps of the search area that includes current estimates that the target resides in each point~\cite{YMP04}. The agent then plans an optimal path within this map with the intent, for example, of optimizing the rate of uncertainty decrease~\cite{KaBe06}. Cooperative searches typically include communication between the agents that can be transmitted up to a given distance, or even without any restriction.   Models have been suggested where agents can communicate by
altering the environment to which other agent then react~\cite{WAYB08}.  Cooperation without
communication has also been explored to some extent~\cite{Ark92} but  the analysis puts no emphasis on
the speed-up of the search process. In addition, to the best of our knowledge, no works exist in this context that put emphasis  on finding nearby targets faster than faraway one. It is important to
stress that in all those engineering works, the issue of whether
the robots know their total number is typically not addressed, as obtaining such
information does not seem to be problematic. Furthermore, in many
works, robots are not identical and have unique identities.

\section{Model and Definitions}\label{sec:preliminaries}

\subsection{General Setting}
We introduce the {\em Ants Nearby Treasure Search (ANTS)}  problem, in which $k$ mobile {\em agents}  are searching for a {\em treasure} on some topological domain. Here, we focus on the two-dimensional plane.
The agents are probabilistic mobile machines (robots).
They are identical, that is, all agents execute the same protocol $\cP$, and their execution differ only due to the outcome of their random coins.
Each agent has some limited field of view, {\em i.e.,} each agent can see its surrounding up to a distance of  some $\varepsilon>0$. Hence, for simplicity, instead of considering the two-dimensional plane, we assume that the agents  are actually walking on the integer two-dimensional infinite grid $G=\mathbb{Z}^2$ (they can traverse an edge of the grid in both directions). The search is central place, that is, all $k$ agents initiate the search from some central node $s\in G$, called the {\em source}. Before the search is initiated,
an adversary locates the treasure  at some node $\varsigma\in G$, referred to as the {\em target} node.
 Once the search is initiated, the agents cannot communicate among themselves.

 The {\em distance} between two nodes  $u,v\in G$, denoted $d(u,v)$, is simply the Manhattan  distance between them, {\em i.e.,} the number of edges on the shortest path connecting $u$ and $v$ in the grid~$G$.  For a node $u$, let $d(u):=d(s,u)$ denote the distance between $u$ and the source node. We denote by $D$ the distance between the source node and the target, {\em i.e.,} $D=d(\varsigma)$.
 It is important to note that the agents have no a-priori information about the location of $\varsigma$ or about $D$.
 We say that the agents {\em find} the treasure when one of the agents visits the target node~$\varsigma$.
 The goal of the agents it to find the treasure as fast as possible as a function of both $D$ and $k$.

Since we are mainly interested in lower bounds, we assume a very liberal setting. In particular,
 we do not restrict either the computational power or  the navigation capabilities of agents. Moreover, we put no restrictions on the internal storage used for navigation. On the other hand, we note that for constructing upper bounds, the algorithms we consider use simple procedures that can be implemented using relatively few resources.
 For example, with respect to navigation, we only  assume  the ability to perform the following basic procedures: (1) choose a
direction uniformly at random, (2) walk in a ``straight line'' to a prescribed distance, (3) perform a {\em spiral search} around a given node\footnote{The spiral search is a particular deterministic search algorithm that starts at a node $v$ and enables the agent  to visit  all nodes at distance  $\Omega(\sqrt{x})$ from $v$ by traversing $x$ edges, for every integer $x$  (see, {\em e.g.,} \cite{BCR91}). For our purposes, since we are interested with asymptotic results only, we can replace this atomic navigation protocol with any procedure that guarantees such a property. For simplicity, in what follows, we assume that for any integer $x$, the spiral search of length $x$ starting at a node $v$ visits all nodes at distance at most $\sqrt{x}/2$ from $v$.},
and (4) return to the source node.

\subsection{Oracles and Advice}
We would like to model  the situation in which before the search actually starts, some information may be exchanged  between the  agents at the source node.
In the context of ants, this preliminary communication may be quite limited. This may be because of  difficulties in the communication that are inherent to the ants or the environment, {\em e.g.,} due to faults or limited memory, or because of asynchrony issues, or simply because agents are identical and it may be difficult for ants to distinguish one agent from the other. Nevertheless,  we consider a very liberal setting in which this preliminary communication  is almost unrestricted.

More specifically, we consider a centralized algorithm called {\em oracle} that assigns advices to agents in a preliminary stage. The oracle, denoted by $\cO$, is a probabilistic\footnote{It is not clear whether  or not  a probabilistic oracle is strictly more powerful than a deterministic one.  Indeed, the oracle assigning the advice is unaware of $D$, and may thus potentially use the randomization to reduce the size of the advices by
 balancing between the efficiency of the search for small values of $D$ and larger values.} centralized algorithm that receives as input a set of $k$ agents and assigns an {\em advice} to each of the $k$ agents.
We  assume that the oracle may use a different protocol for each $k$; Given $k$, the randomized algorithm used for assigning the advices to the $k$~agents is denoted by $\bigO_k$. Furthermore,  the oracle may assign a different advice to each agent\footnote{We note that even though we consider a very liberal setting, and allow a very powerful oracle, the oracles we use for our upper bounds constructions
are very simple and rely on much weaker assumptions. Indeed, these oracles are not only deterministic but also assign the same advice to each of the $k$ agents.}.
Observe that this definition of an oracle enables it to simulate almost any reasonable preliminary communication between the agents\footnote{For example, it can simulate to following very liberal setting. Assume that in the preprocessing stage,   the $k$ agents are organized in a clique topology, and
that each agent can send a separate message to each other agent. Furthermore, even though the agents are identical, in this preprocessing stage, let us assume that agents can distinguish the messages received from different  agents, and that each of the $k$ agents may use a different probabilistic protocol for this preliminary  communication.  In addition, no restriction is made either on  the memory and computation capabilities of agents or on the preprocessing  time, that is, the preprocessing stage takes finite, yet unlimited, time.}.

 It is important to stress that even though all agents execute the same searching protocol, they may start the search with different advices. Hence, since their searching protocol may rely on the content of this initial advice,
 agents with different advices may behave differently. Another important remark  concerns the fact that some part of the advices may be used for encoding (not necessarily disjoint)  identifiers. That is, assumptions regarding the settings in which not all agents are identical and there are several types of agents can be captured by our setting of advice.

To summarize, a {\em search algorithm} is a pair $\langle \cP,\cO\rangle$ consisting of a randomized searching protocol $\cP$ and randomized oracle $\cO=\{\cO_k\}_{k\in \Natural}$. Given
$k$ agents, the randomized oracle $\cO_{k}$ assigns a separate advice to each of the given agents. Subsequently, all agents initiate the actual search by letting each of the  agents  execute protocol $\cP$ and using the corresponding advice as input to $\cP$. Once the search is initiated, the agents cannot communicate among themselves.

\subsection{Complexities}

\paragraph{Advice Complexity.}
Consider  an oracle $\cO$.
Given $k$, let $\Psi_{\cO}(k)$ denote the maximum number of bits devoted for encoding the advice of an agent, taken over all possible advice assignments of $\cO_k$, and over the $k$ agents. In other words, $\Psi_{\cO}(k)$ is the minimum number of bits required for encoding the advice,  assuming that the number of agents is $k$.  
The function $\Psi_{\cO}(\cdot)$ is called the {\em advice size} function of oracle  $\cO$.
(When the context is clear, we may omit the subscript $\cO$ from $\Psi_{\cO}(\cdot)$ and simply use $\Psi(\cdot)$ instead.)

Observe that a lower bound on the size of the advice is also a lower bound on the {\em memory size} of an agent, which is required merely
for storing the output of the preprocessing stage, {\em i.e.,} the initial data-structure required for the search. In other words, even without considering the memory of agents used for the actual navigation, the lower bounds on the advice we establish imply that so many bits of memory must be used initially by the agents.

As mentioned, the advice complexity also directly relates to the {\em state} complexity, {\em i.e.,} the minimum number of states an agent has already when commencing the search. Specifically, a lower bound of $x\geq 1$ bits on the advice size implies a lower bound of $2^{x-1}$ on the state complexity.

Algorithms that do not use any advice (in particular, no information regarding $k$ is
available to agents) are called \emph{uniform}.   (The term ``uniform'' is chosen to stress that
agents execute the same algorithm regardless of their number, see, {\em e.g.,}~\cite{KSV11}.)

\paragraph{Time Complexity.}
When measuring the time to find the treasure, we assume that all computations internal to an agent are performed in zero time. Since we are mostly interested in lower bounds, we assume a liberal setting also in terms of synchronization.
That is, we assume that the execution proceeds in discrete rounds, each taking precisely one unit of time. In each such round, each agent can perform internal computations (that take zero time) and traverse an incident edge. (For upper bound purposes, this synchronization assumption can  be removed if we measure the time according to the slowest edge-traversal.) We also assume that all agents start the search simultaneously at the same round. (For upper bound purposes, this assumption can also be easily removed by starting to count the time when the last agent initiates the search.)

The {\em expected running time} of a search algorithm  $\cA:=\langle \cP,\cO\rangle$ is  the expected time (or number of rounds) until at least one of the agents finds the treasure (that is, the agent is located at the target node $\varsigma$ where the treasure is located).
The expectation  is defined with respect to the coin tosses made by the (probabilistic) oracle $\cO$ assigning the advices to the agents, as well as the subsequent coin tosses made by the agents executing $\cP$. We denote the expected running time of an algorithm $\cA$ by $\tau=\tau_{\cA}(D,k)$.
In fact, for our lower bound to hold, it is sufficient to assume that the probability that the treasure is found by time $2\tau$ is at least $1/2$. By Markov's inequality, this assumption is indeed weaker than the assumption that the expected running time is $\tau$.

Note that if an agent knows $D$, then it can potentially find the treasure in time $\bigO(D)$, by
walking to a distance $D$ in some direction, and then performing a ``circle'' around the source of radius $D$ (assuming, of course, that its navigation abilities enable it  to perform such a circle). With the absence of knowledge about $D$, a single agent can find the treasure in time $\bigO(D^2)$ by performing a spiral search  around the source (see, {\em e.g.,} \cite{BCR91}). The following observation implies that $\Omega(D+D^2/k)$ is a lower bound on the expected running time of any search algorithm. The proof is straightforward but we nevertheless give it here, for completeness purposes.\\

\begin{observation}\label{simple-lower}
The expected running time of any algorithm is $\Omega(D+D^2/k)$, even if the number of agents~$k$ is known to all agents.
\end{observation}
\begin{proof}
To see why the observation holds,  consider the case of $k$ agents, and a search algorithm $\cA$ whose expected running time is $T$. Clearly,  $T\geq D$, because it takes $D$ time to merely reach the treasure.
Assume, towards contradiction, that $T<D^2/4k$.  In any execution of $\cA$, by time $2T$, the $k$ agents can visit a total of at most $2Tk<D^2/2$ nodes. Hence,
 by time  $2T$, more than half of the nodes in $B_D:=\{u\mid 1\leq d(u)\leq D\}$
were not visited. Therefore, there must exist a node
$u\in B_D$  such that the probability that $u$ is visited by time $2T$ (by at least one of the agents) is less than $1/2$. If the adversary locates the treasure at $u$ then the expected time to find the treasure is strictly greater than~$T$, which contradicts the assumption.
\end{proof}
\paragraph{Competitive Analysis.}
We evaluate the time performance of an algorithm  with respect to the lower bound given by Observation~\ref{simple-lower}.
Formally,  let $\Phi(k)$ be a function of $k$. A search  algorithm $\cA:=\langle \cP,\cO\rangle$ is called {\em $\Phi(k)$-competitive}~if
$$\tau_{\cA} (D,k)\leq \Phi(k)\cdot (D+D^2/k),$$ for every integers $k$ and $D$.
Our goal is to establish connections between the {\em size} of the advice, namely $\Psi(k)$, and the competitiveness  $\Phi(k)$ of the search algorithm.

\section{An $\bigO(1)$-Competitive Algorithm with Advice of Size $\lceil\log\log k\rceil$ }\label{sec:upper-opt}

As a warm-up, this section presents an $\bigO(1)$-competitive algorithm with advice  size of $\lceil\log\log k\rceil$ bits.
We start with Theorem \ref{thm:know-k} showing that agents can achieve an asymptotically optimal running time if they know the precise value of $k$. This algorithm requires advice size of $\lceil\log k\rceil$ bits, that are used to encode the value of $k$ at each agent. However, as explained in the following corollary, Theorem \ref{thm:know-k} also holds if agents merely know a constant approximation of $k$. As this knowledge can be encoded using $\lceil\log\log k\rceil$ bits, we obtain an $\bigO(1)$-competitive algorithm with advice  size of $\lceil\log\log k\rceil$ bits.

 The $\bigO(1)$-competitive algorithm that we consider is based on performing the following operations iteratively:
\begin{enumerate}
\item
walk in a ``straight line'' to a point chosen uniformly at random in some ball around the source,  
\item
perform a spiral search around that point for a  length which is proportional to the size of the ball divided by the number of agents, and 
\item
return to the source. 
\end{enumerate}
More formally, we
consider the following algorithm, termed $\cA_{\kk}$ (see also Figure~\ref{fig:phasei} for an illustration).
\renewcommand{\baselinestretch}{1.2}
\begin{algorithm}
Each agent performs the following double loop\;
\For(the \emph{stage} $j$ defined as follows){$j$ from $1$ to $\infty$}{
  \For(the \emph{phase} $i$ defined as follows){$i$ from $1$ to $j$}{
  \begin{smallitemize}
      \item
  Go to a node $u$ chosen uniformly at random among the nodes in $B_i\coloneqq \{u\mid d(u) \leq 2^i\}$\;
\item
Perform a spiral search for time $t_i\coloneqq2^{2i+2}/ k$\;
\item
Return to the source $s$\;
  \end{smallitemize}
  }
 }

\caption{The non-uniform algorithm $\cA_{\kk}$.}\label{alg:non-uniform}
\end{algorithm}

\begin{figure}[ht]
    \begin{center}
        \includegraphics[height=6cm]{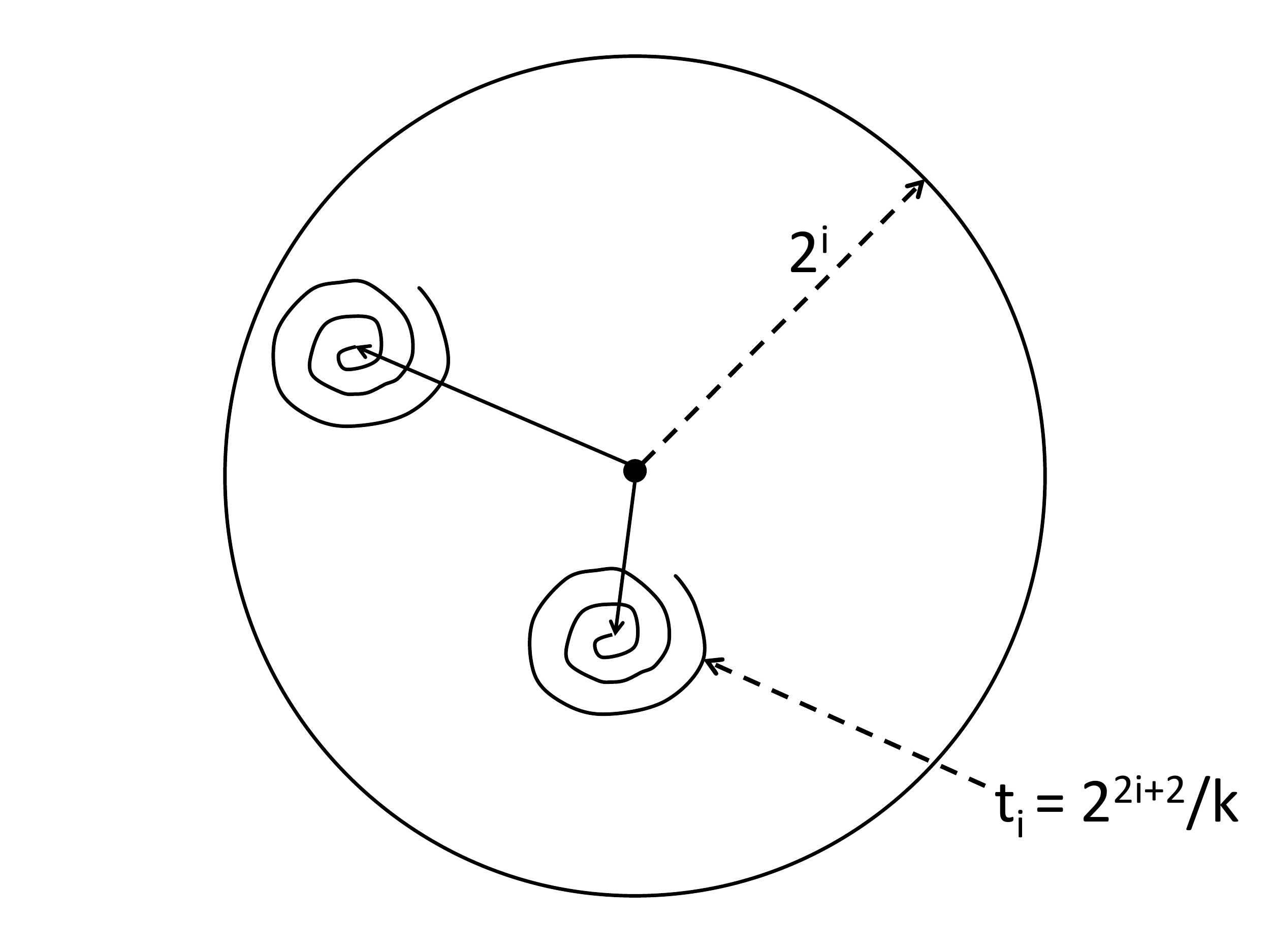}
    \end{center}
    \caption{Illustration of two agents performing phase $i$.}\label{fig:phasei}
\end{figure}

Note that Algorithm $\cA_{\kk}$ assumes that agents know their total number, $k$. The next theorem states that its performances are asymptotically optimal.\\
\begin{theorem}\label{thm:know-k}
  Algorithm $\cA_{\kk}$ is $\bigO(1)$-competitive.
\end{theorem}

\begin{proof}
Assume that the treasure is located at some node $\varsigma$ at distance $D$ from the source. We now show that the expected time to find it is $\bigO(D+D^2/k)$.

Denote $s\coloneqq\lceil\log D\rceil$. A phase $i$, for which $i\geq s$, is called {\em large}. 
Note that the ball $B_i$ corresponding to a large phase $i$ contains the node $\varsigma$ holding the treasure.
 Fix a positive integer $\ell$ and consider the time $T_\ell$ until each agent completed $\ell$
large phases.
Each time an agent performs a large phase $i$, the agent finds the
treasure if the chosen node $u$ belongs to
the ball $B(\varsigma, \sqrt{t_i}/2)$ around $\varsigma$, the node holding the treasure (this follows from the definition of a spiral search, see Footnote 3). 
Since $\varsigma\in B_i$, we get that  at least some constant fraction of the ball $B(\varsigma, \sqrt{t_i}/2)$
is contained in $B_i$.
The probability of choosing a node $u$ in that fraction is thus
$\Omega(|B(\varsigma,\sqrt{t_i}/2)|/|B_i|)=\Omega(t_i/ 2^{2i})$, which is at least $\beta/k$ for
some positive constant $\beta$.
Thus, the probability that by time $T_\ell$ none of the $k$ agents finds the treasure (while
executing their respective $\ell$ large phases $i$) is at most
$(1-\beta/k)^{k\ell}$, which is at most $\gamma^{-\ell}$ for some
constant $\gamma$ greater than $1$.

For $i\in\mathbf{N}$, let $\psi(i)$ be the time required to execute a phase $i$.
Note that $\psi(i)=\bigO(2^i+2^{2i}/k)$. Hence, the time until all agents
complete stage $j$ for the first time is
\[\sum_{i=1}^j \psi(i)=\bigO\left(\sum_{i=1}^j 2^i + 2^{2i}/k\right)=\bigO(2^j+2^{2j}/k).\]
This equation implies that for any integer $r$, all agents
complete their respective stages $s+r$ by time
$\hat{T}(r)\coloneqq \bigO(2^{s+r}+2^{2(s+r)}/k)$. Observe that by this time, all agents
have completed at least $\ell=r^2/2$ large phases $i$. Consequently, as explained before, the probability that none of the $k$ agents finds the treasure by time
$\hat{T}(r)$  is at most
$\gamma^{-\ell}=\gamma^{-r^2/2}$. Hence, the expected running time is at most
\begin{align*}
\cT_{\cA_{\kk}}(D,k)&=\bigO\left(\sum_{r=1}^\infty
\frac{2^{s+r}}{\gamma^{r^2/2}}+\frac{2^{2(s+r)}}{k\gamma^{r^2/2}}\right)\\
&=\bigO\left(2^{s}+2^{2s}/k\right)=\bigO(D+D^2/k).
\end{align*}
This establishes the theorem.
\end{proof}

Fix a constant $\rho\ge1$. We say that the agents have a \emph{$\rho$-approximation} of $k$, if,
initially, each agent $a$ receives as input a value $k_a$ satisfying
$k/\rho\le k_a\le k\rho$.\\

\begin{corollary}\label{corr:opt}
Fix a constant $\rho\ge1$.
    Assume that the agents have a $\rho$-approximation of their total number. Then,
    there exists a (non-uniform) search algorithm $\cA^*$ that is $\bigO(1)$-competitive.
\end{corollary}

\begin{proof}
    Each agent $a$ executes Algorithm $\cA_{\kk}$ (see the proof of Theorem~\ref{thm:know-k}) with the parameter
    $k$ equal to $k_a/\rho$. By the definition of a $\rho$-approximation,
    the only difference between this case and the case where the agents know
    their total number,
    is that for each agent, the time required to perform each spiral search
    is multiplied by a constant factor of at most $\rho^2$.
    Therefore, the analysis in the proof of Theorem~\ref{thm:know-k}
    remains essentially the same and the running time is increased by a multiplicative factor of at most
    $\rho^2$.
\end{proof}
\noindent The following theorem directly follows from the previous corollary.\\
\begin{theorem}\label{cor:known}
There exists an $\bigO(1)$-competitive search algorithm using advice size of $\lceil \log \log k\rceil$ bits.
\end{theorem}

\section{Lower Bounds on the Advice}\label{sec:lower}
In Section \ref{sec:weaker-lower}  we give an informal proof that the competitiveness of any uniform search algorithm cannot be $o(\log k)$.  The proof provides some important ingredients of the main lower bound proof. In Section~\ref{sec:definition} we provide definitions and preliminary observations required for formally proving our lower bounds. In Section \ref{sec:uniform} we establish a  stronger result than the one mentioned in the informal discussion, stating that  the competitiveness of any uniform search algorithm is in fact strictly more than logarithmic. This result is a special case of the main lower bound theorem,  namely, Theorem \ref{main-lower}, that we prove next, in  Section \ref{sec:main-lower}.   In the following section, namely, Section~\ref{sec:specific-lower}, we prove other lower bound results that directly follow from Theorem \ref{main-lower} as corollaries.

\subsection{An Informal Proof of a Weak Lower Bound for Uniform Algorithms}\label{sec:weaker-lower}
We now provide an informal proof explaining why there are no uniform search algorithms which are $\Phi(k)$-competitive, where $\Phi(k)=o(\log k)$ and $\Phi(k)$ is non-decreasing. So assume that the running  time is $(D + \frac{D^2}{k})\cdot \Phi(k)$, and let us show that   $\Phi(k)=\Omega(\log k)$.

We first fix some large number $K$ that will serve as an upper bound on the number of agents.
Let $T=K\cdot\Phi(K)$, and assume that the treasure is actually placed very far away, specifically, at distance $D>2T$. Note that this means, in particular,
that by time $2T$ the treasure has not been found yet\footnote{Although the treasure is placed far, in principle, it could have also been placed much closer to the source. In fact,  had the treasure been placed at a distance of roughly the number of agents $k<K$,  the algorithm would have had enough time $T$ to find it. This is the intuition behind of choice of the value of parameter $T$.}.

We consider ${\log K}/{2}$ rings $R_i$ centered at the source $s$ that span the region containing all points at distance between roughly $\sqrt{K}$ and roughly $K$ from the source. See Figure \ref{fig:Rings} below.
Specifically, let $j_0=\lfloor {\log K}/{2}\rfloor-2$, and consider the set of integers $I=\{1,2,\cdots, j_0\}$. For $i\in I$,  define the ring
  \[
  R_i=\{u\mid 2^{j_0+i}<d(u,s)\leq 2^{j_0+i+1}\}.
  \]
\begin{figure}[ht]
    \begin{center}
        \includegraphics[width=11cm]{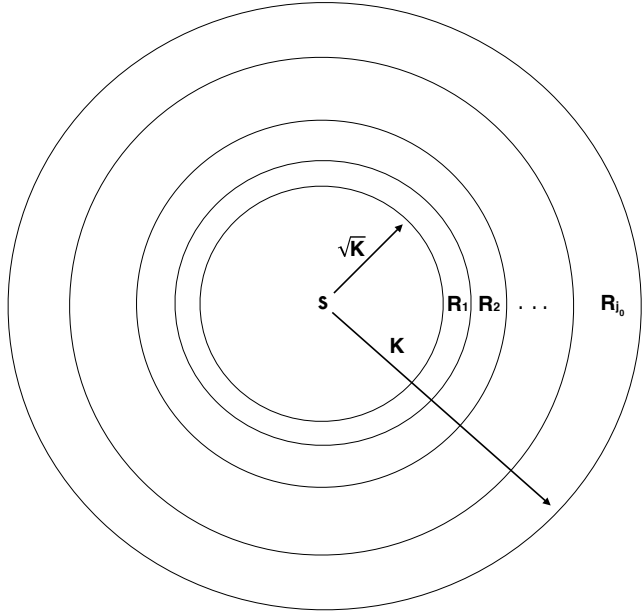}
    \end{center}
    \caption{Dividing the relevant region to rings $R_i$.}\label{fig:Rings}
\end{figure}
\noindent Fix  $i\in I$ and let us now assume that the number of agents is $k_i=2^{2i}$. Note that $k_i\leq K$.
Further, observe that the number of points in $R_i$ is
\[|R_i|=\Theta(2^{2j_0+2i})=\Theta(K\cdot k_i).
\] 
We also have $k_i\leq 2^{j_0+i} < d(u,s)$, for every $u\in R_i$. That is, any point in $R_i$ is at distance at least $k_i$ from the source. Hence, the expected time until the $k_i$ agents cover each $u\in R_i$  is at most 
\[
\left(d(u,s)+\frac{d(u,s)^2}{k_i}\right) \cdot\Phi(k_i)~\leq~ 2\frac{d(u,s)^2}{k_i} \cdot\Phi(k_i)~\leq~ K\cdot\Phi(K)~=~T.
\]
Recall that in a uniform algorithm, all agents start the search with the same state.
Thus, the expected number of nodes in $R_i$ that an agent visits by time $T$ is $\Omega(|R_i|/k_i)= \Omega(K)$. Observe,  this holds  for all $i\in I$.
Hence,  the expected number of nodes  that a single agent visits
by time $T$ is  $\Omega(K\cdot |I|)=\Omega(K\cdot \log K)$. On the other hand, this number is also at most $T$ because an agent can visit at most one node in one unit to time. As $T=K\cdot\Phi(K)$, this implies that the competitiveness is 
 $\Phi(K)=\Omega(\log K)$, as required.

\subsection{Definitions and Preliminary Observations}\label{sec:definition}
We now give preliminary definitions  that will be used to formally state and prove our lower bounds.
We begin with the following definition. A non-decreasing function $\Phi(x)$ is called {\em relatively-slow} if $\Phi(x)$  is sublinear ({\em i.e.,} $\Phi(x)=o(x)$) and if
there exist two positive constants $c_1$ and $c_2<2$ such that when
restricted to $x>c_1$, we have $\Phi(2x)<c_2\cdot \Phi(x)$. Note that this definition captures many natural sublinear functions\footnote{For example, note that the functions of the form $\alpha_0+\alpha_1\log^{\beta_1} x+\alpha_2\log^{\beta_2}\log x+ \alpha_3 2^{\log^{\beta_3 }\log x} \log x+\alpha_4\log^{\beta_4} x\log^{\beta_5}\log x$, (for
non-negative constants $\alpha_i$ and $\beta_i$, $i=1,2,3,4,5$ such that $\sum_{i=1}^4\alpha_i>0$)
 are all relatively-slow.}.
 
 We will consider different scenarios for exponentially increasing number of agents. Specifically, for every integer $i$, set $$k_i=2^i.$$ Given a  relatively-slow function $\Phi(x)$, let $i_0$ be the smallest integer such that for every integer $i\geq i_0$, we have $k_i>\Phi(k_i)$. (The existence of $i_0$ is guaranteed by the fact that $\Phi$ is sublinear). 
Let $T\geq 2^{2i_0}$ be a sufficiently large integer (in fact, we will typically take $T$ to infinity). Without loss of generality, we assume that $\log T$ is an even number. In all the lower bound proofs, we assume that the treasure is placed somewhere at  distance $$D:= 2T+1~.$$ Note that this means, in particular,
that by time $2T$ the treasure has not been found yet. Next, define the set of integers $$I= \left[i_0,\frac{\log T}{2}\right]~.$$
Fix an integer $i\in I$ and set
\[
\quad d_i\coloneqq \left\lfloor \sqrt{\frac{T\cdot k_i}{ \Phi(k_i)}} \right\rfloor.
\]
Further, let $B(d_i)$ denote the ball of radius $d_i$ around the source node, {\em i.e.,}$$B(d_i):=\{v\in G: d(v)\leq d_i\}~. $$ 
Note that the largest $k_i$ we consider is $\sqrt{T}$. Therefore, in the case that $i_0$ is small, the sizes of the balls of $B(d_i)$ range from roughly $T$ to roughly $T\cdot\frac{\sqrt{T}}{\Phi(\sqrt{T})}$.

For every set of nodes $S\subseteq B(d_i)$, let $\chi_i(S)$ denote the random variable
indicating the number of nodes in $S$ that were visited by at least one of the
$k_i$ agents by time $4T$.  (For short, for a singleton node $u$, we write $\chi_i(u)$ instead of  $\chi_i(\{u\})$.) Note that the value of  $\chi_i(S)$ depends on the values of the coins tosses made by the oracle for assigning the advices as well as on the values of the coins tossed by  the $k_i$ agents.
Finally, define the ring $$R_i :=B(d_i)\setminus B(d_{i-1})~.$$

\noindent The following claim states that the number of nodes in Ring $R_i$ is roughly $d_i^2$.\\
\begin{claim}
For each integer $i\in I$, we have $\Omega(|R_i|)=\Omega \left(\frac{T\cdot k_i}{\Phi(k_i)}\right)=\Omega(d_i^2).$
\end{claim}
\begin{proof}
The claim follows from the following sequence of equations
\begin{align*}
|R_i|=\Omega\left(d_{i-1}(d_i-d_{i-1})\right)=\Omega\left(\sqrt{\frac{T\cdot k_{i-1}}{\Phi(k_{i-1})}}\cdot \left( \sqrt{\frac{T\cdot k_{i}}{\Phi(k_{i})}} - \sqrt{\frac{T\cdot k_{i-1}}{\Phi(k_{i-1})}} \right)\right)=\\
\Omega\left(\frac{T\cdot k_{i-1}}{\Phi(k_{i-1})}\cdot \left( \sqrt{ \frac{2\Phi(k_{i-1})}{\Phi(k_{i})}} - 1 \right)\right)=\Omega \left(\frac{T\cdot k_{i-1}}{\Phi(k_{i-1})}\right)=\Omega \left(\frac{T\cdot k_i}{\Phi(k_i)}\right)=\Omega(d_i^2)~,
\end{align*}
where the fourth and fifth equalities follow from the facts that $\Phi$ is relatively-slow.
\end{proof}

\noindent We now claim  that by time $4T$, the $k_i$ agents are expected to cover a constant fraction of Ring $R_i$.\\
\begin{claim}
For each integer $i\in I$, we have $\EE(\chi_i(R_i))=\Omega (|R_i|)~.$
\end{claim}
\begin{proof}  
To see why the claim holds, observe that since $i\in I$, and $k_i> \Phi(k_i)$, we have 
\[ k_i\leq 2^{\log T/2} = \sqrt{T}\leq d_{i}~,\]
which implies that for every node $u\in B(d_i)$ the expected time to visit $u$ is at most  $$\tau(d_i,k_i)\leq 2\frac{d_i^2\Phi(k_i)}{k_i}\leq 2T~.$$
Thus, by Markov's inequality, the probability that $u$ is visited by time $4T$ is at least $1/4$,
that is,~$\PP(\chi(u)= 1)\geq 1/4$. Hence, $\EE(\chi(u))\geq 1/4$.
By linearity of expectation, $\EE(\chi(R_i))=\sum_{u\in R_i} \EE(\chi(u))\geq |R_i|/4$.
\end{proof}

\noindent The following claim directly follows from the previous two claims.\\
\begin{claim}\label{claim:uniform}
For each integer $i\in I$, we have $$\EE(\chi_i(R_i))=\Omega \left(\frac{T\cdot k_i}{\Phi(k_i)}\right)~.$$
\end{claim}

\subsection{No Uniform Search Algorithm  is $\bigO(\log k)$-Competitive}\label{sec:uniform}
In this subsection we prove a lower bound for the competitiveness of uniform algorithm. The lower bound is stronger than the one mentioned in Section \ref{sec:weaker-lower}. The proof follows similar logic as in the weaker version, but uses the more refined choosing of parameters, as defined in Section \ref{sec:definition}. \\

\begin{theorem}\label{cor:lowerbound}
There is no uniform search algorithm that is $\bigO(\log k)$-competitive.
\end{theorem}

\noindent The theorem follows directly from the following lemma by setting $\Phi(k)=\bigO(\log k)$. \\

\begin{lemma}\label{cor:unifrom}
Consider a relatively-slow function $\Phi$ such that $\sum_{j=1}^\infty 1/\Phi(2^j)=\infty$.
There is no uniform search
algorithm that is $\Phi(k)$-competitive. 
\end{lemma}
\begin{proof}
Suppose for contradiction that there exists a uniform search algorithm with running time less
than $\tau(D,k)\coloneqq \Phi(k)\cdot (D+D^2/k)$, where $\Phi$ is relatively-slow. 
Our goal is to show that $\sum_{j=1}^\infty 1/\Phi(2^j)$ converges.

We  use the assumption that the algorithm is uniform. In this case, all agents are not only executing the same algorithm but also start with the same state. Hence,  Claim \ref{claim:uniform} implies that
for every integer $i\in I$, the expected number of
nodes in $R_i$ that each agent visits by time $4T$ is $\Omega\left(\frac{T}{\Phi(k_i)}\right)$.
Since the rings $R_i$ are pairwise disjoint, the linearity of expectation
implies that the expected number of nodes that an agent visits by time $4T$ is
\[
\Omega\left(\sum_{i=i_0}^{\log
T/2}\frac{T}{\Phi(k_i)} \right)=T\cdot \Omega\left(\sum_{i=i_0}^{\log
T/2}\frac{1}{\Phi(2^i)} \right).
\]
On the other hand, this number is at most $4T$ because an agent can visit at most one node in each time unit. Consequently, $\sum_{i=i_0}^{\log T/2}\frac{1}{\Phi(2^i)}$ must converge as $T$
goes to infinity.
This completes the proof of Lemma~\ref{cor:unifrom}.
\end{proof}

\subsection{The Main Lower Bound Theorem}\label{sec:main-lower}

Here we present the main theorem of this section, namely, Theorem \ref{main-lower}. The theorem is a generalization of Lemma \ref{cor:unifrom}. 
The statement of  the theorem uses  definitions as given in Section \ref{sec:definition}.\\

\begin{theorem}\label{main-lower}
There exist constants $i_0$ and $c$ such that the following holds. Let $k$ denote the number of agents. For every $\Phi(k)$-competitive search algorithm using advice of size $\Psi(k)$, where $\Phi(\cdot)$ is relatively-slow  and  $\Psi(\cdot)$ is non-decreasing, and for every sufficiently large integer $T$, we have:
 $$\frac{1}{2^{\Psi(T)}}\sum_{i=i_0}^{\log T} \frac{1}{\Phi(2^i)}<c~.$$
\end{theorem}
\bigskip
\begin{remark}
Note that in the case of uniform algorithms, we have advice size being zero, {\em i.e.,} $\Psi(\cdot)\equiv 0$, and therefore the requirement that the sum in Theorem \ref{main-lower} converges is the same as the requirement in Lemma \ref{cor:unifrom}. Naturally, as the advice size grows, the competitiveness $\Phi(\cdot)$ may decrease without violating the inequality. 
\end{remark}

\medskip
\begin{proof}[Proof (of Theorem \ref{main-lower})]
Recall that  Claim \ref{claim:uniform} holds for any algorithm, whether it is uniform or not. Therefore, for each integer $i\in I$,  the expected number of
nodes in $R_i$  visited by at least one of the $k_i$ agents by time $4T$ is: $$\EE(\chi_i(R_i))=\Omega \left(\frac{T\cdot k_i}{\Phi(k_i)}\right)~.$$

\noindent We now consider the case when running with $k_i$ agents after being assigned advice by the oracle. Let $W(j, i)$ denote the random variable indicating the number of nodes in $R_i$ visited by the $j$'th agent by time $4T$.
It follows  that for every integer $i\in I$, we have:
$$
\sum_{j=1}^{k_i} \EE(W(j,i))~=\EE\left(\sum_{j=1}^{k_i} W(j,i)\right)~\geq~\EE(\chi_i(R_i))~=~ \Omega \left(\frac{T\cdot k_i}{\Phi(k_i)}\right).
$$
By the pigeon-hole principle,  there exists an integer $j\in \{1,2,\cdots, k_i\}$   for which
$$
\EE(W(j,i))=\Omega(T/\Phi(k_i)).
$$
\noindent In other words, in expectation, the $j'$th agent in the group of $k_i$ agents covers at least $\Omega(T/\Phi(k_i))$ nodes in $R_i$ by time $4T$.

Let  $\ADV$ denote the set of all possible advice used for agents, in all the scenarios with $k_i$ agents, where $i\in I$. Note that the advice given by the oracle
to any of the $k_i$ agents must use at most $\Psi(k_i)\leq \Psi(\sqrt{T})$ bits. For every fixed integer $x$, 
the number of different possible advices that can be encoded on precisely $x$ bits is~$2^x$, therefore  the total number of different possible advice used
is $\sum_{x=1}^{\Psi(\sqrt{T})} 2^x\leq 2^{\Psi(\sqrt{T})+1}$. In other words, we have:
\begin{claim}\label{claim:ADV}
$|\ADV|\leq 2^{\Psi(\sqrt{T})+1}$.
\end{claim}

\noindent Now, for each advice  $a\in \ADV$, let $M(a,i)$ denote the random variable indicating the number of nodes in $R_i$ that an agent with advice $a$ visits by time $4T$. Note that the value of $M(a,i)$ depends only  on the values of the coin tosses made by the agent.
On the other hand,
note that the value of $W(j,i)$ depends on the results of the coin tosses made by the oracle assigning the advice, and the results of the coin tosses made by the agent that uses the assigned advice.
Recall,
the oracle may assign an advice to agent $j$ according to a distribution that is different than the distributions used for other agents. However, regardless of the distribution used by the oracle for agent $j$, it must be the case that there exists an advice $a_i\in \ADV$, for which $
\EE(M(a_i
,i))\geq \EE(W(j,i))$. Hence, we obtain:
$$
\EE(M(a_i
,i))=\Omega(T/\Phi(k_i)).
$$
\noindent In other words, in expectation, an agent with advice $a_i$ covers at least $\Omega(T/\Phi(k_i))$ nodes in $R_i$ by time $4T$.

 Let $A=\{a_i
\mid  i\in I\}$.
Consider now an ``imaginary'' scenario\footnote{The scenario is called imaginary, because,  instead of letting the oracle assign the advice
for the agents, we impose a particular advice to each agent, and let the agents perform the search with our advices. Note that even though such a scenario cannot occur by the definition of the model, each individual agent with advice $a$ cannot distinguish this case from the case that the number of agents was some $k'$ and the oracle assigned it the advice $a$.} in which  we execute the search algorithm with $|A|$ agents, each having a different advice in~$A$. That is, for each advice $a\in A$, we have a different agent executing the algorithm using advice $a$.
 For every set $S$ of nodes in the grid $G$, let $\hat{\chi}(S)$ denote the random variable
indicating the number of nodes in $S$ that were visited by at least one of these
$|A|$ agents (in the ``imaginary'' scenario) by time $4T$. Let $\hat{\chi}:=\hat{\chi}(G)$ denote the random variable
indicating the total number of nodes that were visited by at least one of these
agents by time $4T$.

By definition, for each $i\in I$, the expected number of nodes in $R_i$ visited by at least one of these $|A|$ agents is at least the expected number of nodes in $R_i$ visited by the particular agent with advice $a_i$. This number is $\EE(M(a_i,i))$. Hence, we obtain:
$$\EE(\hat{\chi}(R_i))\geq \EE(M(a_i
,i))=\Omega(T/\Phi(k_i)).$$
Since the sets $R_i$ are pairwise disjoint, the linearity of expectation
implies that the expected number of nodes covered by these agents by time $4T$ is
$$
\EE(\hat{\chi})~\geq~\sum_{i=i_0}^{\frac{1}{2}\log
T}  \EE(\hat{\chi}(R_i))   ~=~  \Omega\left(\sum_{i=i_0}^{\frac{1}{2}\log
T}\frac{T}{\Phi(k_i)} \right) $$  $$=T\cdot \Omega\left(\sum_{i=i_0}^{\frac{1}{2}\log
T}\frac{1}{\Phi(2^i)}\right).
$$
Observe that $A\subseteq \ADV$, and therefore, by Claim \ref{claim:ADV}, we have $|A|\leq 2^{\Psi(\sqrt{T})+1}$. Hence, once more by linearity of expectation,    there must exist an advice $\hat{a}\in A$, such that the expected number of nodes that an agent with advice $\hat{a}$ visits by time $4T$ is
$$
T\cdot \Omega\left(\sum_{i=i_0}^{\frac{1}{2}\log
T}\frac{1}{\Phi(2^i)\cdot 2^{\Psi(\sqrt{T})}} \right).
$$
Since each agent may visit at most one node in one unit of time, it follows that,  for every $T$ large enough, the sum $$\frac{1}{ 2^{\Psi(\sqrt{T})}}\sum_{i=i_0}^{\frac{1}{2}\log T} \frac{1}{\Phi(2^i)}$$ is at most some fixed constant.  The proof of the theorem now follows
by replacing the variable $T$ with~$T^2$.
\end{proof}

\subsection{More Lower Bounds}\label{sec:specific-lower}
We next describe important lower bounds results that directly follow from Theorem \ref{main-lower} as corollaries.
We first consider the case of uniform algorithm, in which  the advice size is zero, that is, $\Psi(k)=0$.

\subsubsection{Any $\bigO(\log k)$-Competitive Algorithm Requires Advice of Size $\log\log\log k-\bigO(1)$}
We have established that there is no uniform algorithm that is $\bigO(\log k)$ competitive. We now show that, in fact, $\bigO(\log k)$-competitive
algorithms require advice of at least $\log\log\log k-\bigO(1)$ bits.\\
\begin{theorem}\label{cor:logk}
There is no $\bigO(\log k)$-competitive search algorithm that uses advice of size $\log\log\log k-\omega(1)$.
\end{theorem}
\begin{proof}
Assume that the competitiveness is  $\Phi(k)=\bigO(\log k)$. Since $\Phi(2^i)=\bigO(i)$, we have
$$\sum_{i=1}^{\log k} 1/\Phi(2^i)=  \sum_{i=1}^{\log k} 1/ i= \Omega(\log\log k).$$
According to Theorem~\ref{main-lower}, the sum $$\frac{1}{2^{\Psi(k)}}\sum_{i=1}^{\log k} 1/\Phi(2^i)=\Omega(\log\log k/2^{\Psi(k)})$$ must converge as $k$ goes to infinity.
 In particular, we cannot have $\Psi(k)=\log\log\log k-\omega(1)$.
\end{proof}
\subsubsection{The Minimal Advice Size of Algorithms with Sub-Logarithmic Competitiveness}
We next show that the minimal advice size quickly grows to $\Omega(\log\log k)$, as the competitiveness becomes sub-logarithmic.
In particular, we show that $\bigO(1)$-competitive algorithms require advice size of $\log\log k- \bigO(1)$ bits. This result shows that the upper bound of $\log \log k+\bigO(1)$ bits mentioned in Theorem \ref{cor:known} is in fact tight.\\

\begin{theorem}\label{cor:lower-advice}
Consider  any  $\Phi(k)$-competitive search algorithm using advice of size $\Psi(k)$.  Assume that $\Phi(\cdot)$ is relatively-slow  and that $\Psi(\cdot)$ is non-decreasing.
Then, ${\Psi(k)}=\log\log k - \log \Phi(k) - \bigO(1)$. In particular:
\begin{itemize}
\item
Any $\bigO(1)$-competitive search algorithm  must have advice  size $\Psi(k)= \log\log k -\bigO(1)$.
\item
Let $\epsilon<1$ be a positive constant. Any $\bigO(\log^{1-\epsilon} k)$-competitive search algorithm must have advice  size $\Psi(k)= \epsilon\log\log k -\bigO(1)$.
\item
Let $\epsilon<1$ be a positive constant. Any $\bigO(\frac{\log k}{2^{\log^\epsilon\log k}})$-competitive search algorithm requires advice of size
 $\Psi(k)= \log^{\epsilon}\log k-\bigO(1)$.
\end{itemize}
\end{theorem}

\begin{proof}
Theorem~\ref{main-lower} implies that for every $k$, we have: $$\frac{1}{2^{\Psi(k)}}\sum_{i=1}^{\log k} \frac{1}{\Phi(2^i)}=\bigO(1).$$ On the other hand, since $\Phi$ is non-decreasing, we have:
$$\sum_{i=1}^{\log k} \frac{1}{\Phi(2^i)}\geq \frac{\log k}{\Phi(k)}.$$
Hence,
$$ \frac{\log k}{2^{\Psi(k)}\cdot\Phi(k)}=\bigO(1).$$
The main part of the theorem therefore follows: $$\Psi(k)=\log\log k - \log \Phi(k) - \bigO(1).$$
The first item follows by taking $\Psi(k)=\bigO(1)$,  the second item follows by taking $\Psi(k)=\bigO(\log^{1-\epsilon} k)$ and the third item follows by taking $\Psi(k)=\bigO(\frac{\log k}{2^{\log^\epsilon\log k}})$.
\end{proof}

\section{Further Upper Bounds}\label{sec:upper-bound}

In Theorem \ref{cor:known} we showed that there exists an $\bigO(1)$-competitive algorithm that uses $\log\log k +1$ bits of advice. In Theorem \ref{cor:lower-advice} we proved that this bound is tight. In this section we give  two additional tight upper bounds. In Section \ref{sec:upper-uniform} we show that the lower bound for the competitiveness of uniform algorithms mentioned in Theorem \ref{cor:unifrom} is tight. In Section \ref{sec:upper-log} we show that the lower bound of $\log\log\log k$ for the advice size of $\bigO(\log k)$-competitive algorithms mentioned in Theorem \ref{cor:logk} is tight.

\subsection{A Tight Upper Bound for the Competitiveness of Uniform Algorithms}\label{sec:upper-uniform}

We now turn our attention to the case of uniform algorithms. Consider a non-decreasing function   $\Phi\colon\Natural\to\Natural$  such that $\sum_{j=1}^\infty 1/\Phi(2^j)<\infty$. We suggest Algorithm $\cA_{\uniform}$ below as an optimal algorithm for the uniform case.

\renewcommand{\baselinestretch}{1.1}
\begin{algorithm}

Each agent performs the following\;
\For(the epoch $\ell$){$\ell$ from $0$ to $\infty$}{
\For(the stage $i$){$i$ from $0$ to $\ell$}{
\For(the phase $j$){$j$ from $0$ to $i$}{
\begin{smallitemize}
\item $k_j\longleftarrow2^j$\;
\item $D_{i,j}\longleftarrow  \sqrt{2^{(i+j)}/\Phi(2^j)}$
\item
Go to the node $u$ chosen uniformly at\break random among the nodes
in $B(D_{i,j})$\;
\item
Perform a spiral search starting at $u$\break for time $t_{i,j}\coloneqq2^{i+2}/\Phi(2^j)$\;
\item
Return to the source\;
\end{smallitemize}
  }
 }
}

\caption{The uniform algorithm $\cA_{\uniform}$.}\label{alg:u-alg}
\end{algorithm}
\noindent The structure of Algorithm $\cA_{\uniform}$  finds resemblance with the structure of Algorithm $\cA_{\kk}$ described in Section \ref{sec:upper-opt}. Indeed, the basic operations are the same: Go to a node chosen uniformly at random in some ball around the source,  perform a spiral search of some length, and return to the source. In contrast to Algorithm $\cA_{\kk}$ that needs the value of $k$ (or some approximation of it) to specify the length of the spiral search, Algorithm $\cA_{\uniform}$ decides the length of the spiral search based only on the loops indices. 

The theorem below implies that the lower bound mentioned in Theorem \ref{cor:unifrom} is tight.\\

\begin{theorem}\label{thm:comp}
Consider a non-decreasing function   $\Phi\colon\Natural\to\Natural$  such that $\sum_{j=1}^\infty 1/\Phi(2^j)<\infty$.
Algorithm $\cA_{\uniform}$ is a  uniform search
algorithm which is  $\Phi(k)$-competitive.
\end{theorem}

\begin{proof}
Assuming $\Phi(\cdot)$ as given in the theorem, let
us analyze the performance  of Algorithm $\cA_{\uniform}$ and show that its expected
running time is $T(D,k)\coloneqq\Phi(k)\cdot(D+D^2/k)$.
We first note that it suffices to prove the statement
when $k\le D$.
Indeed, if $k>D$, then we may consider only $D$ agents among
the $k$ agents and obtain an upper bound on the running time of $T(D,D)$,
which is less than $2T(D,k)$.

\begin{assert}\label{as:1}
For every integer $\ell$, the time until all agents complete epoch $\ell$ is $\bigO(2^\ell)$.
\end{assert}
\noindent For the assertion to hold, it is sufficient to prove that stage $i$
in epoch $\ell$ takes time $\bigO(2^{i})$.
To this end, notice that phase~$j$ takes time $\bigO(D_{i,j}+2^{i}/\Phi(2^j))$ which is at most $\bigO(2^{(i+j)/2}+2^{i}/\Phi(2^j))$.
Therefore, stage $i$ takes time
\[
\bigO\left(\sum_{j=0}^{i}(2^{(i+j)/2}+2^{i}/\Phi(2^j))\right)=\bigO(2^{i}).
\]
This establishes Assertion~\ref{as:1}.

Set $s\coloneqq\lceil\log((D^2\cdot \Phi(k))/k)\rceil+1$. Recall that we consider the case where $D\geq k$. Therefore $2^s\geq D^2\cdot \Phi(k))/k \geq k$. 
In particular, there exists $j\in\{0,\ldots,s\}$ such that $2^{j}\le
k<2^{j+1}$.\\
\begin{assert}\label{as:2}
There exists a constant $c\in(0,1)$ such that
the probability that none of the agents  finds the treasure while executing  phase $j$ of stage
$i$ is at most $c$.
\end{assert}
\noindent To see this, first note that the treasure is inside the ball $B(D_{i,j})$, for every $i\geq s$.
Indeed, since $i\geq s$,  we have
$$D_{i,j}=\sqrt{\frac{2^{i+j}}{\Phi(2^j)}}\ge\sqrt{\frac{2^{s+j}}{\Phi(2^j)}}\geq \sqrt{\frac{D^2\cdot \Phi(k)}{k}\cdot \frac{2^{j+1}}{\Phi(2^j)}}>D.$$
The total number of nodes in the ball $B(D_{i,j})$ is
$\bigO(D_{i,j}^2)=O(2^{i+j}/\Phi(2^j))$,
and at least a third of the ball of radius $\sqrt{t_{i,j}}$ around the
treasure is contained in $B(D_{i,j})$.
Consequently, the probability for an agent $a$ to choose a node $u$ in
a ball of radius $\sqrt{t_{i,j}}$ around the treasure in phase $j$ of stage
$i$ is
\[
\Omega\left(t_{i,j}/|B(D_{i,j})|\right)=\Omega\left(\frac{2^i/\Phi(2^j)}{2^{i+j}/\Phi(2^j)}\right)=\Omega\left(2^{-j}\right).
\]
If this event happens, then the treasure is found during the corresponding spiral search of agent $a$.
As a result, there exists a positive constant $c'$ such that
the probability that none of the $k$ agents finds the treasure
during phase $j$ of stage $i$ is at most
$(1-c'\cdot 2^{-j})^k\le(1-c'\cdot 2^{-j})^{2^j}\le e^{-c'}$.
This establishes Assertion~\ref{as:2}.

By the time that all agents have completed their respective
epoch $s+\ell$, all agents have performed
$\Omega(\ell^2)$ stages $i$ with $i\ge s$. By Assertion~\ref{as:2},
for each such $i$, the probability that the treasure is not found
during stage  $i$ is at most $c$ for some constant $c<1$.
Hence, the probability that the treasure is not found
during any of those $\Omega(\ell^2)$ stages is at most $1/d^{\ell^2}$ for some
constant $d>1$.
Assertion~\ref{as:1} ensures that all agents complete epoch $s+\ell$ by time $\bigO(2^{s+\ell})$,
so the expected running time is
$\bigO(\sum_{\ell=0}^\infty 2^{s+\ell}/d^{\ell^2})=\bigO(2^s)=\bigO(D^2\cdot \Phi(k)/k)$, as desired.
\end{proof}

\noindent Setting $\Phi(x)\coloneqq \lceil \log^{1+\varepsilon} x\rceil$ yields the following corollary.\\
\begin{corollary}\label{cor:comp}
For every positive constant $\varepsilon$, there exists a uniform search
algorithm that is $\bigO(\log^{1+\varepsilon}k)$-competitive.
\end{corollary}

\subsection{Upper Bounds on the Advice Size of (Sub-)Logarithmic Competitive Algorithms}\label{subsec:upper-bound}
Let $\Psi(k)$ be some increasing function such that: $\log\log\log k\leq \Psi(k)\leq \log\log k$. In this section we show that advice of size $\bigO(\Psi(k))$ enables to obtain an $\bigO\left(\frac{\log k}{2^{\Psi(k)}}\right)$-competitive algorithm. This will show that the lower bounds on the advice size mentioned in Theorem \ref{cor:lower-advice}
are asymptotically tight. Before proving this result, we first show, in Section \ref{sec:upper-log},  that  advice of size $\log\log\log k+\bigO(1)$ allows to construct an $\bigO(\log k)$-competitive algorithm. Establishing this latter result shows that the lower bound  mentioned in Theorem~\ref{cor:logk} is tight. 

Recall that Algorithm $\cA^*$ from Corollary \ref{corr:opt} is an $\bigO(1)$-competitive algorithm that assumes the knowledge of a 2-approximation to $k$ and runs in time $\bigO(D+D^2/k)$.  
Both upper bounds mentioned in this section are based on a simple iterative simulation of the Algorithm $\cA^*$  for different guesses of $k$. The guesses are constructed from the advice. As in all our non-uniform upper bounds, all $k$ agents receive the same advice that, in fact, serves as some approximation to $k$. The overhead in running time, or the competitiveness, corresponds to the number of guesses that the algorithm needs to check. This number can be estimated directly from the quality of the advice as an approximation to $k$.

\subsubsection{Using $\log\log\log k+\bigO(1)$ Bits of Advice to Obtain an $\bigO(\log k)$-Competitive Algorithm}\label{sec:upper-log}

We describe an $\bigO(\log k)$-competitive algorithm, called $\cA_{\log k}$, that uses $\Psi(k)=\log\log\log k+\bigO(1)$  bits of advice. 
Without loss of generality, we may assume that  $k$ is sufficiently large, specifically, $k\geq 4$. Given $k$ agents, the oracle simply encodes the advice $\bigO_k=\lfloor\log\log k\rfloor$ at each agent. Note that since $k\geq 4$, we have $\bigO_k\geq 1$. Observe also that the advice $\bigO_k$ can be encoded using $\log\log\log k+\bigO(1)$ bits.

Given an advice $\alpha$, each agent first constructs $S(\alpha)$ -- the set of integers $k$ such that
$k$ is a power of~2 and $O_k=\alpha$. That is:
$S(\alpha)=\{2^i \mid i\in \Natural \mbox{~and~} {2^\alpha}\leq i < {2^{\alpha+1}}\}.$
\noindent In particular, for all $k\in S(\alpha)$, we have $\lfloor\log\log k\rfloor=\alpha$.
It then enumerates the elements in $S(\alpha)$ from small to large, namely, $S(\alpha)=\{ k_{1}(\alpha),k_{2}(\alpha),\cdots, k_{|S(\alpha)|}(\alpha)  \},$ where $k_{i}(\alpha)$ is the $i$'th smallest integer in $S(\alpha)$, for $i\in\{1,2, \cdots, |S(\alpha)|\}$. 
The following two observations describe two properties of $S(\alpha)$. The first observation states that $S(\alpha)$ contains an element which is a 2-approximation of the true number of agents $k$, and the second observation states that the number of elements in $S(\alpha)$ is logarithmic. \\ 
\begin{observation}\label{obs:2-approx}
There exists $k_{i^*}(\alpha)\in S(\alpha)$, such that  $k_{i^*}(\alpha)\leq k< 2k_{i^*}(\alpha)$.
\end{observation}
\bigskip
\begin{observation}\label{obs:log}
$|S(\alpha)|=\bigO(\log k)$.
\end{observation}
\noindent The proof of Observation \ref{obs:2-approx} is straightforward. To see why Observation \ref{obs:log} holds, let $k_{\max}$ denote the maximum value such that $\bigO_{k_{\max}}=\alpha$, and let $k_{\min}$ denote the minimum value such that $\bigO_{k_{\min}}=\alpha$.  We know, $S(\alpha)\leq \log k_{\max}$.
Since $\bigO_{k_{\max}}=\bigO_{k_{\min}}$, we have $\log\log k_{\max}<\log\log k_{\min}+1$, and hence $\log k_{\max}<2\log k_{\min}$. The  observation follows as $k_{\min}\leq k$.

Recall that Algorithm $\cA^*$ from Corollary \ref{corr:opt} is an $\bigO(1)$-competitive algorithm that assumes the knowledge of a 2-approximation to $k$ and runs in time $\bigO(D+D^2/k)$. The above observations imply that given the advice $\alpha$,  there are $\bigO(\log k)$ options to guess a 2-approximation for $k$. Indeed, these guesses are the elements in $S(\alpha)$. The algorithm  $\cA_{\log k}$ we construct simply loops over all these guesses, executing Algorithm $\cA^*$  for a limited number of steps, in each iteration, assuming the number of agents in the current guess. This incurs an overhead of $\bigO(\log k)$ (number of guesses) over the optimal competitiveness of Algorithm $\cA^*$. More precisely,
the following algorithm is executed by an agent with advice $\alpha$.

\renewcommand{\baselinestretch}{1.1}
\begin{algorithm}

Construct $S(\alpha)=\{2^i \mid i\in \Natural \mbox{~and~} {2^\alpha}\leq i < {2^{\alpha+1}}\}$\\
Enumerate the elements in $S(\alpha)$, that is, $S(\alpha)=\{ k_{1}(\alpha),k_{2}(\alpha),\cdots, k_{|S(\alpha)|}(\alpha) \}$\\
\For(the \emph{stage} $j$ ){$j$ from $1$ to $\infty$}{
  \For(the \emph{phase} $i$ ){$i$ from $1$ to $|S(\alpha)|$}{
  \begin{smallitemize}
  \item
Execute $2^j$ steps of Algorithm $\cA^*$ assuming the number of agents is $k_i(\alpha)$.
 \item
Return to the source $s$. 
  \end{smallitemize}
  }
 }

\caption{Execution of $\cA_{\log k}$ by an agent with advice $\alpha$.}\label{alg:k-alg}
\end{algorithm}

\bigskip

%

\begin{theorem}\label{th:upperbound} 
Algorithm $\cA_{\log k}$ is   an $\bigO(\log k)$-competitive algorithm that uses $\log\log\log k+\bigO(1)$ bits of advice.
\end{theorem}
\begin{proof}
 The advice given to each agent is $\alpha=\lfloor\log\log k\rfloor$, which as mentioned, can be encoded using $\log\log\log k+\bigO(1)$ bits. Our goal is therefore to show that $\cA_{\log k}$ is $\bigO(\log k)$-competitive. 
 
By Observation~\ref{obs:2-approx}, there exists $k_{i^*}(\alpha)\in S(\alpha)$, such that  $k_{i^*}(\alpha)\leq k< 2k_{i^*}(\alpha)$. 
A {\em correct} step is a step executed in phase $i^*$ (for any stage). In such steps, agents execute a step 
of the non-uniform algorithm $\cA^*$ assuming the number of agents is $k_{i^*}(\alpha)$. Note  that in each stage $j$, all agents execute $2^j$ consecutive correct steps. This means that the expected number of correct steps until one of the agents finds the treasure is at most twice the expected running time of Algorithm 
$\cA^*$, which is $\bigO(D+D^2/k)$, by Corollary~\ref{corr:opt}. By time $T\log k$, Algorithm $\cA_{\log k}$ makes 
$\Omega(T)$ correct steps. It follows that the expected running time of Algorithm $\cA_{\log k}$ is $\bigO((D+D^2/k)\cdot \log k)$, or in other words, the competitiveness is $\bigO(\log k)$.
\end{proof}
\subsubsection{Upper Bounds on the Advice Size of Sub-Logarithmic Competitive Algorithms}\label{sec:sublinear}
 
In this subsection we prove the following theorem which implies that the lower bounds on the advice size mentioned in Theorem \ref{cor:lower-advice}
are asymptotically tight.\\
\begin{theorem}\label{thm:Psi}
Let $\Psi(k)$ be a non-decreasing function such that: $$\log\log\log k\leq \Psi(k)\leq \log\log k~.$$ There exists an $\bigO\left(\frac{\log k}{2^{\Psi(k)}}\right)$-competitive algorithm with advice of size $\bigO(\Psi(k))$.

\end{theorem}
\begin{proof}
Recall from Subsection \ref{sec:upper-log} that there exists an integer $C$ such that $\log\log\log k+C$ bits of advice allow each agent to extract $S(\alpha)$, which is a set with  $\log k$ elements, one of which is a 2-approximation of $k$. Let us denote this element of $S(\alpha)$ by $k_{i^*}(\alpha)$.  

Our desired algorithm, called $\cA_{\Psi}$, assigns the same advice to all $k$ agents. This  advice, termed $\alpha_{\Psi(k)}$, contains precisely $2B(k)$ bits, where $B(k)=\Psi(k)+C$. The first half of the advice, encapsulating precisely $B(k)$ bits, is used to  encode $S(\alpha)$. This can be done since $\Psi(k)\geq \log\log\log k$ and as we know $\log\log\log k+C$ bits suffice to encode this information\footnote{Since we may have more bits than we need here, we can just pad zero's as the most significant bits.}.
The next $B(k)$ bits in each advice are used to specify\footnote{This procedure is straightforward. Simply follow a binary search in $S(\alpha)$ for the element $k_{i^*}(\alpha)$, up to depth $B(k)$.} a subset $S(\alpha,\Psi)$ of $S(\alpha)$ of size $\bigO\left(\frac{\log k}{2^{B(k)}}\right)$ which contains the element $k_{i^*}(\alpha)$. 

Given its advice $\alpha_{\Psi(k)}$, an agent can easily construct $S(\alpha)$ using the first half of the advice, and then construct the subset $S(\alpha,\Psi)$ using the second half of the advice. At this point, the agent doesn't know which element in $S(\alpha,\Psi)$ is $k_{i^*}(\alpha)$ but it knows that such an element exists. 
The way the agent proceeds  is very similar to how it would operate in $\cA_{\log k}$. That is, it iteratively tries guesses for a limited number of steps in each iteration.
The only difference is that instead of looping over the set $S(\alpha)$ that contains $\log k$ elements, it loops over the smaller set $S(\alpha,\Psi)$, that contains only $\bigO\left(\frac{\log k}{2^{B(k)}}\right)=\bigO\left(\frac{\log k}{2^{\Psi(k)}}\right)$ elements. More formally, given $S(\alpha,\Psi)$, an agent operates as follows:
\renewcommand{\baselinestretch}{1.1}
\begin{algorithm}

Enumerate the elements in $S(\alpha,\Psi)=\{k_1(\alpha,\Psi),k_2(\alpha,\Psi),\cdots\}$  \\
\For(the \emph{stage} $j$ ){$j$ from $1$ to $\infty$}{
  \For(the \emph{phase} $i$ ){$i$ from $1$ to $|S(\alpha,\Psi)|$}{
  \begin{smallitemize}
  \item
Execute $2^j$ steps of Algorithm $\cA^*$ assuming the number of agents is $k_i(\alpha,\Psi)$.
 \item
Return to the source $s$. 
  \end{smallitemize}
  }
 }

\caption{Execution of $\cA_{\Psi}$ by an agent after extracting $S(\alpha,\Psi)$.}\label{alg:k-alg-f}
\end{algorithm}

\noindent The analysis of Algorithm $\cA_{\Psi}$ is similar to the analysis of Algorithm $\cA_{\log k}$ (see the proof of Theorem \ref{th:upperbound}). The only difference is that the overhead of $\cA_{\Psi}$ corresponds to the size of $S(\alpha,\Psi)$ rather than to the size of $S(\alpha)$. Since the size of $S(\alpha,\Psi)$ is at most $\bigO\left(\frac{\log k}{2^{\Psi(k)}}\right)$ we  obtain that the competitiveness of $\cA_{\Psi}$ is $\bigO\left(\frac{\log k}{2^{\Psi(k)}}\right)$, as desired.
\end{proof}

\begin{remark}
The advice used in Algorithm $\cA_{\Psi}$ contains $2\Psi(k)+O(1)$ bits. The first half of the advice is used to encode the set of candidate guesses $S(\alpha)$ and the second to reduce this set by a factor of $2^\Psi(k)$. In fact, one could save on the advice size. 
Indeed, encoding $S(\alpha)$ requires only $\log\log\log k+O(1)$ bits. This suggests that the advice could potentially be reduced to  $\Psi(k)+\log\log\log k+O(1)$ bits, while maintaining the $\bigO\left(\frac{\log k}{2^{\Psi(k)}}\right)$ competitiveness. Note, 
however, that to extract $S(\alpha)$ using the first part of the advice, the agent would need to know when the first part of the advice ends. This information could be encoded in the advice, using a separate field of size $\lceil \log\log\log\log k\rceil$ bits, but again one would need to encode where this field ends. Continuing this way recursively would show that the advice can be reduced  from $2\Psi(k)+O(1)$ to $\Psi(k)+\log^{<3>} k+\log^{<4>} k +\cdots + \log^{<i>} k +\cdots + O(\log^* k)$ bits, where $\log^{<i>} k$ stands for applying the logarithm function $i$ times where the first time is applied on $k$.
\end{remark}

\subsection{Harmonic Search}

The algorithms described earlier in the paper are relatively simple but still require the use of
non-trivial iterations, which may be complex for simple and tiny agents, such as ants.
If we relax the requirement of bounding the expected running time and demand only that the treasure
be found with some low constant probability, then it is possible to avoid one
of the loops of the algorithms. However, a sequence of iterations still needs
to be performed.

In this section, we propose an extremely simple algorithm, coined
the \emph{harmonic search algorithm}\footnote{The name harmonic was chosen
because of structure resemblances to the celebrated harmonic algorithm for the
$k$-server problem, see, {\em e.g.,}~\cite{BaGr00}.}, which does {\em not} perform in iterations
and is essentially composed of three steps:
\begin{smallitemize}
\item
 Choose a random direction and walk in this direction for a distance of $d$, where $d$ is 
chosen randomly according to a distribution in which the probability of
choosing $d$  is roughly inverse proportional to $d$,
\item Perform a local search ({\em e.g.,} a spiral search) for time roughly $d^2$, and
\item Return to the source.
\end{smallitemize}
It turns out that this extremely simple algorithm has a good probability of quickly finding the treasure,
if the number of agents is sufficiently large.

More specifically,  the algorithm depends on a positive constant parameter $\delta$ that is fixed in
advance and governs the performance of the algorithm.
For a node $u$, let $p(u)\coloneqq\frac{c_\delta}{d(u)^{2+ \delta}},$ where $c_\delta$ is
the normalizing factor, defined so that
$\sum_{u\in V(G)} p(u)=1$, where $V(G)$ is the set of nodes in the grid.  
The $\delta$-harmonic search algorithm is given below. 

\renewcommand{\baselinestretch}{1.2}
\begin{algorithm}
Each agent performs the following three actions\;
1. Go to a node $u\in V(G)$ with probability $p(u)$\;
2. Perform a spiral search for time $t(u)\coloneqq d(u)^{2+\delta}$\;
3. Return to the source.
\caption{The $\delta$-harmonic search algorithm.}\label{alg:h-alg}
\end{algorithm}
\noindent The following theorem provides bounds on the performances of the harmonic algorithm. The proof of the theorem essentially follows using arguments similar to those introduced earlier, {\em e.g.,} in the proofs of
Theorems~\ref{thm:know-k} and~\ref{thm:comp}. \\
\begin{theorem}
Let  $0<\delta$ and $\epsilon<1$ be two arbitrarily small constants. 
There exists a positive real number $\alpha$ such that
if  $k>\alpha  D^{\delta}$, then with probability at least $1-\varepsilon$, the $\delta$-harmonic algorithm finds the treasure in time  $\bigO(D+\frac{D^{2+\delta}}{k}).$
\end{theorem}

\begin{proof}
Let  $0<\delta$ and $\epsilon<1$ be two arbitrarily small constants. Without loss of generality we may assume that $\delta<1$. 
Fix a real number $\beta$ greater than $\ln(1/\varepsilon)$, so
$e^{-\beta}<\epsilon$. Set $\alpha=40\beta/c_\delta$, where $c_\delta$ is the normalization factor. Note that $\alpha$ depends on both $\epsilon$ and $\delta$. 
We assume that the number of agents $k$ is greater than $\alpha D^{\delta}$ and our goal is to show that
with probability of at least $1-\epsilon$, the running time is $\bigO(D+\frac{D^{2+\delta}}{k}).$

For the sake of analysis, let us define $$\lambda := \frac{4\beta D^{1+\delta}}{c_\delta k}< \frac{D}{10}~,$$and   
consider the ball $B_{\lambda}$ of radius $\sqrt{\lambda D}/2$ around the treasure, {\em i.e.,}
$$B_\lambda:=\{u\in G: d(u, \varsigma)\leq \sqrt{\lambda D}/2\}~. $$ 
The next claim  follows by the triangle inequality (no attempt has been made to optimize constants).

\begin{claim}\label{claim:ball}
 For every node $u\in B_{\lambda}$ we have: $3D/4<d(u)<5D/4~.$
\end{claim}

\noindent The following claim implies, in particular, that the probability of finding the treasure is at least as large as the probability of reaching a node in ${B}_{\lambda}$ in step 1 of the algorithm.
\begin{claim}
If an agent goes to a node $u\in {B}_{\lambda}$ in
step 1 of the algorithm, then this agent finds the treasure in step 2.
\end{claim}
\noindent To see why the claim holds, recall that by the definition of a spiral search, if an agent  performs a spiral search from $u$ of length $x$ then it finds all nodes at distance
at most $\sqrt{x}/2$ from $u$.
Now, if $u\in B_{\lambda}$, then the distance from $u$ to the treasure is at most twice the radius of $B_\lambda$, that is, $d(u,\varsigma)<\sqrt{\lambda D}$, since both $u$ and $\varsigma$ are in $B_{\lambda}$. 
It follows that  an agent which performs a spiral search from $u\in B_{\lambda}$
finds the treasure by time $4\lambda D$. So the claim would follow once we show that the length of the spiral search that such an agent performs in step 2 is at least  $$4\lambda D< \frac{2 D^{2}}{5}~. $$ This follows since the length is $$t(u)=d(u)^{2+\delta}>\left(\frac{3D}{4}\right)^{2+\delta}>(3/4)^3\cdot D^2>\frac{2 D^{2}}{5}~.$$\\

\noindent Since each node in ${B}_{\lambda}$ is at distance less than $5D/4$ from the
source (Claim \ref{claim:ball}), it follows that if one of the agents goes to some $u\in B_\lambda$ in step 1 of the algorithm, then the total running time of the algorithm
is $\bigO(D(5/4+\lambda))$, which is $\bigO(D+\frac{D^{2+\delta}}{k})$. Hence, it remains to analyze the probability that at least one of the $k$ agents goes to a node in ${B}_{\lambda}$ in step 1 of the algorithm.

Since $d(u)<5D/4$ for each node $u\in {B}_{\lambda}$, the probability $p(u)$ that a single agent
goes to $u$ in step 1 is at least $\frac{c_\delta}{(5D/4)^{2+ \delta}}\geq \frac{c_\delta}{2D^{2+ \delta}}$ as $\delta<1$.
Since there are at least $\lambda D/2$ nodes in ${B}_{\lambda}$, the probability that a single agent
goes to a node in ${B}_{\lambda}$ in step 1 is at least
\[
\sum_{u\in B_{\lambda}}p(u)\ge|B_\lambda|\cdot\frac{c_\delta}{2D^{2+ \delta}}\ge
\frac{c_\delta\lambda}{4D^{1+\delta}} = 
\frac{\beta}{k}~.
\]
It follows that the probability that no agent goes to a node in $B_{\lambda}$ in step 1 of the algorithm
is at most
\[
 \left(1-\frac{\beta}{k}\right)^k\le e^{-\beta}<\varepsilon.
\]
The theorem follows.
\end{proof}

\section{Discussion}\label{sec:conclusion}
 Distributed systems are abundant on all scales of biology. Providing quantitative descriptions of the complex interactions between the system's components remains as one of the major challenges towards a better understanding of the biological world. For this purpose, biologists are constantly seeking new tools.

The discipline of systems biology~\cite{Kitano} aims to address this challenge by taking a more holistic (non-reductionist) approach to biological research. The mathematical tools used therein are often drawn from physics and control theory and heavily rely on differential equation techniques and on computer simulations. The field of theoretical distributed computing  is also concerned with  complex and dynamical systems. However, the techniques and perceptions it applies are different, being more algorithmically oriented and proof based. It is therefore natural to ask whether, and to what extent, techniques from the field of distributed computing could be applied to biological contexts.

Our results are an attempt to draw non-trivial bounds on biological systems from the field of distributed computing.
We consider a central search setting such as the one performed by desert ants. As stated above, a central place allows for a preliminary stage in which ants may assess some knowledge about their total number.
Our theoretical lower bounds may enable us to relate group search performance to the extent of information sharing within the nest.  Such relations are particularly interesting since they provide  not a single bound but  relations between key parameters. In our case these would be the number of states an agent has when commencing the search  and collective search efficiency.  Combining our lower bounds and measurements of search speed with varying numbers of searchers may provide quantitative evidence regarding the  number of memory bits (or, alternatively, the number of states) immediately as the ants exit the nest to commence the search. Note that such lower bounds would be very difficult to obtain using traditional methodologies. Hence,
obtaining lower bounds of even few bits using this indirect methodology would already demonstrate 
the power of this framework.

 It is important to stress that the complexity of the biological world often defies a theoretical description and collective central place foraging is no exception. Closing the gap between experiment and theory in this case requires further efforts on both the theoretical and the experimental fronts. Theoretically, one must employ more realistic models (\emph{e.g.} not all searchers exit the nest at once) and then solve them to provide precise (non-asymptotic) bounds. 
 On the experimental side, as stated above, there are a several candidate social species (desert ants, honeybees) that are solitary foragers and possess the cognitive capabilities for complex navigation. Since the navigational capabilities of these insects depend on factors such as landmarks and the sun, experiments must be conducted in the field. Furthermore, to fit the specification of our model, the conditions should be such that food is scarce and located at an unpredictable location.
   It has been shown that, in these conditions, the foraging range of honey bees increases with colony size \cite{Beekman}. Although this is consistent with the search algorithms described here one must verify experimentally that search times are indeed sped up when the number of worker increases.  This can be done by placing a bait\footnote{The bait itself may be ``imaginary'' in the sense that it need not be placed physically on the terrain. Instead, a camera can just be placed filming the location, and detecting when searchers get nearby.} at a certain distance from the nest (hive) entrance and measuring the finding times for different distances and different colony sizes. Our results indicate that if the speed-up of the search time is near-linear in the number of searchers then the properties of the search pattern of the single insect must depend on the number of available foragers. Additional experimental validations can be performed by measuring and quantifying the characteristics of single insect trajectories.  This further theoretical and experimental work is beyond the scope of this paper.
Nevertheless, our results provide a ``proof of concept'' for such a methodology.

\bigskip

\paragraph{Acknowledgements.}  The authors are thankful to  Jean-S\'ebastien Sereni and Zvi Lotker for helpful discussions, that were particularly useful for facilitating the proofs of Theorems \ref{thm:know-k} and \ref{cor:lowerbound}. In addition, the authors are thankful to the anonymous reviewers for helping to improve the presentation of the paper, and for providing us with an idea that was used to prove the upper bound in Section \ref{sec:sublinear}.

\end{document}